\documentclass{lmcs} 
\pdfoutput=1
\usepackage[utf8]{inputenc}

\usepackage{lastpage}
\lmcsdoi{21}{1}{27}
\lmcsheading{}{\pageref{LastPage}}{}{}%
{Apr.~07,~2023}{Mar.~20,~2025}{}


\usepackage{hyperref}
\theoremstyle{plain} 

\usepackage{xspace}
\usepackage{centernot}

\input{macros.sty}

\begin{document}

\title[Inclusion for Boundedly-Ambiguous VAS is Decidable]{Language Inclusion for Boundedly-Ambiguous\texorpdfstring{\\}{} Vector Addition Systems is Decidable}

\thanks{This paper is a full version of~\cite{DBLP:conf/concur/CzerwinskiH22}
which appeared at the CONCUR 2022 conference.
Both authors are supported by the INFSYS ERC grant, agreement no. 950398.}

\author[W. Czerwi\'nski]{Wojciech Czerwi\'nski\lmcsorcid{0000-0002-6169-868X}}[]
\author[P. Hofman]{Piotr Hofman\lmcsorcid{0000-0001-9866-3723}}[]

\address{University of Warsaw, Poland}

\email{wczerwin@mimuw.edu.pl, piotr.hofman@uw.edu.pl}

\begin{abstract}
We consider the problems of language inclusion and language equivalence for Vector Addition Systems with States (VASS)
with the acceptance condition defined by the set of accepting states (and more generally by some upward-closed conditions).
In general, the problem of language equivalence is undecidable even for one-dimensional VASS, 
thus to get decidability we investigate restricted subclasses. 
On the one hand, we show that the problem of language inclusion of a VASS in $k$-ambiguous VASS (for any natural $k$)
is decidable and even in Ackermann. On the other hand, we prove that the language equivalence problem
is already Ackermann-hard for deterministic VASS.
These two results imply Ackermann-completeness for language inclusion and equivalence in several possible restrictions.
Some of our techniques can be also applied in much broader generality in infinite-state systems, namely for some subclass of
well-structured transition systems.
\end{abstract}

\maketitle

\section{Introduction}
Vector Addition Systems (VAS) together with almost equivalent Petri Nets and Vector Addition Systems with States (VASS)
are one of the most fundamental computational models with many applications in practice for modelling concurrent behaviour.
There is also an active field of theoretical research on VAS, with a prominent example
being the reachability problem whose complexity was recently established to be \ackermann-complete~\cite{DBLP:conf/focs/Leroux21,DBLP:conf/focs/CzerwinskiO21}
and \cite{DBLP:conf/lics/LerouxS19}.
An important type of questions that can be asked for any pair of systems is whether they are equivalent
in a certain sense. The problem of language equivalence (acceptance by configuration) was already proven
to be undecidable in 1975 by Araki and Kasami~\cite{DBLP:journals/tcs/ArakiK76} (Theorem 3).
They also have shown that the language equivalence (acceptance by configuration) for deterministic VAS
is reducible to the reachability problem, thus decidable, as the reachability problem was shown to be decidable by Mayr
a few years later in 1981~\cite{DBLP:conf/stoc/Mayr81}.
The equality of the reachability sets of two given VAS was also shown undecidable in the 70-ties
by Hack~\cite{DBLP:journals/tcs/Hack76}.
Jan\v{c}ar has proven in 1995 that the most natural behavioural equivalence, namely the bisimilarity equivalence
is undecidable for VASS~\cite{DBLP:journals/tcs/Jancar95}. His proof works for only two dimensions (improving the previous
results~\cite{DBLP:journals/tcs/ArakiK76}) and is applicable also to language equivalence (this time as well for acceptance
by states). A few years later in 2001 Jan\v{c}ar has shown in~\cite{DBLP:journals/tcs/Jancar01}
that any reasonable equivalence in-between language equivalence (with acceptance by states)
and bisimilarity is undecidable (Theorem 3) and \ackermann-hard even for systems with finite reachability set (Theorem 4).
For the language equivalence problem, the state-of-the-art was improved a few years ago.
In~\cite{DBLP:conf/lics/HofmanMT13} (Theorem 20) it was shown that already for one-dimensional VASS
the language equivalence (and even the trace equivalence, namely language equivalence with all the states accepting) is undecidable. 

As the problem of language equivalence (and similar ones) is undecidable for general VASS (even in very small dimensions)
it is natural to search for subclasses in which the problem is decidable. The decidability of the problem
for deterministic VASS~\cite{DBLP:journals/tcs/ArakiK76,DBLP:conf/stoc/Mayr81} suggests that restricting
nondeterminism might be a good idea. Recently a lot of attention was drawn
to unambiguous systems~\cite{DBLP:conf/dcfs/Colcombet15},
namely systems in which each word is accepted by at most one accepting run, but can potentially have many non-accepting runs.
Such systems are often more expressive than the deterministic ones however they share some of their good properties, for example 
\cite{DBLP:journals/ijfcs/CadilhacFM13}.
In particular many problems are more tractable in the unambiguous case than in the general nondeterministic case.
This difference is already visible for finite automata. The language universality and the language equivalence problems
for unambiguous finite automata are in \nctwo~\cite{DBLP:journals/ipl/Tzeng96} (so also in \ptime) while
they are in general \pspace-complete for nondeterministic finite automata.
Recently it was shown that for some infinite-state systems the language universality, equivalence and inclusion problems
are much more tractable in the unambiguous case than in the general one. 
There was a line of research
investigating the problem for register automata~\cite{DBLP:conf/stacs/MottetQ19,DBLP:conf/stacs/BarloyC21,
DBLP:conf/icalp/CzerwinskiMQ21} culminating in the work of Boja\'nczyk, Klin and Moerman~\cite{DBLP:conf/lics/BojanczykKM21}.
They have shown that for unambiguous register automata with guessing the language equivalence problem is in \exptime
(and in \ptime for a fixed number of registers).
This result is in a sheer contrast with the undecidability of the problem in the general case even for two register automata
without guessing~\cite{DBLP:journals/tocl/NevenSV04} or one register automata with guessing (the proof can be obtained
following the lines of~\cite{DBLP:journals/tocl/DemriL09} as explained in~\cite{DBLP:conf/icalp/CzerwinskiMQ21}).
Recently, it was also shown in~\cite{DBLP:conf/concur/CzerwinskiFH20} that the language universality problem
for VASS accepting with states is \expspace-complete in the unambiguous case in contrast to \ackermann-hardness
in the nondeterministic case (even for one-dimensional VASS)~\cite{DBLP:conf/rp/HofmanT14}.

\myparagraph{Our contribution}
In this article we follow the line of~\cite{DBLP:conf/concur/CzerwinskiFH20} and consider problems of language equivalence
and inclusion for unambiguous VASS and also for their generalisations $k$-ambiguous VASS (for $k \in \N$) in which each word
can have at most $k$ accepting runs.
The acceptance condition is defined by some upward-closed set of configurations which generalises a bit the acceptance
by states considered in~\cite{DBLP:conf/concur/CzerwinskiFH20}.
Notice that the equivalence problem can be easily reduced to the inclusion problem,
so we prove lower complexity bounds for the equivalence problem and upper complexity bounds for the inclusion problem.

Our only lower bound result is the following one; it is proven by a rather straightforward reduction from the reachability problem for VASS.

\begin{thm}\label{thm:ackermann-hardness}
The language equivalence problem for deterministic VASS is \ackermann-hard.
\end{thm}

Our main technical result concerns VASS with bounded ambiguity.
\begin{thm}\label{thm:complement-bavass}
There is an algorithm that for each $k \in \N$ and a $k$-ambiguous VASS $V$ constructs, in elementary time,
a VASS with a downward-closed set of accepting configurations which recognises the complement of the language of $V$.
\end{thm}

Its rather simple consequence is the following result.
\begin{thm}\label{thm:k-ambiguous}
For each $k \in \N$, the language inclusion problem of a VASS in a $k$-ambiguous VASS is in \ackermann.
\end{thm}
Briefly speaking to prove Theorem~\ref{thm:k-ambiguous} we use Theorem~\ref{thm:complement-bavass}
and check non-emptiness of the intersection of the language of the first VASS with the complement of the language of the second VASS. The problem reduces to the reachability problem in VASS,
which is in \ackermann~\cite{DBLP:conf/lics/LerouxS19}.
In Section~\ref{sec:bavass} we present the details of the proof. The special case of Theorem~\ref{thm:k-ambiguous} is the following statement.

\begin{thm}\label{thm:unambiguous}
The inclusion problem of a nondeterministic VASS language in an unambiguous VASS language is in \ackermann.
\end{thm}

Despite the fact that Theorem~\ref{thm:unambiguous} is a special case of Theorem~\ref{thm:k-ambiguous}
we decided to present separately the proof of Theorem~\ref{thm:unambiguous} because it introduces a different technique of independent interest.
Additionally, it is a good introduction to a more technically challenging proof of Theorem~\ref{thm:k-ambiguous}.

The proof of Theorem~\ref{thm:unambiguous} uses a novel technique but is quite simple.
We add a regular lookahead to a VASS and use results on regular separability of VASS
from~\cite{DBLP:conf/concur/CzerwinskiLMMKS18} to reduce the problem, roughly speaking, to the deterministic case. This technique can be applied to more general systems namely well-structured transition-systems~\cite{DBLP:journals/tcs/FinkelS01}.
We believe that it might be interesting on its own and reveal some
connection between separability problems and the notion of unambiguity.

The proof of Theorem~\ref{thm:complement-bavass} proceeds in three steps. First, we show that the problem for
$k$-ambiguous VASS can be reduced to the case when the control automaton of the VASS is $k$-ambiguous.
Next, we show that the control automaton can even be made $k$-deterministic
(roughly speaking, for each word there are at most $k$ runs).
Finally, we show that for a VASS with $k$-deterministic control automaton one can compute in elementary
time a VASS with a downward-closed set of accepting configurations recognising its complement.
We emphasise that to prove our main technical result, namely Theorem~\ref{thm:complement-bavass}, we
do not need most of our contribution from Section~\ref{sec:uvass}. In fact, the only used part of Section~\ref{sec:uvass} is the construction of the word and automata decoration.

On a way to show Theorem~\ref{thm:complement-bavass} we also prove several other lemmas and theorems,
which we believe may be interesting on their own.
Theorems~\ref{thm:ackermann-hardness}~and~\ref{thm:k-ambiguous} together easily imply the following corollary.

\begin{cor}\label{cor:equivalence}
The language equivalence problem is \ackermann-complete for:
\begin{itemize}
  \item deterministic VASS
  \item unambiguous VASS
  \item $k$-ambiguous VASS for any $k \in \N$
\end{itemize}
\end{cor}

\myparagraph{Organisation of the paper}
In Section~\ref{sec:prelim} we introduce the necessary notions. Then in Section~\ref{sec:detvass}
we present results concerning deterministic VASS.
First, we show Theorem~\ref{thm:ackermann-hardness}.
Next, we prove that the inclusion problem of a VASS language in a language of 
a deterministic VASS, a \kcdeterministic\ VASS or a VASS with holes (to be defined) is in \ackermann.
This is achieved by reducing to the VASS reachability problem.
In Section~\ref{sec:uvass} we define adding a regular lookahead to VASS.
Then we show that with a carefully chosen lookahead we can reduce the inclusion problem of a VASS language
in an unambiguous VASS language into the inclusion problem of a VASS language in the language of deterministic VASS with holes.
The latter is in \ackermann due to Section~\ref{sec:detvass} so the former is also in \ackermann.
In Section~\ref{sec:bavass} we present the proof of Theorem~\ref{thm:complement-bavass} which is our most technically involved contribution.
We also use the idea of a regular lookahead and the result proved in Section~\ref{sec:detvass} about \kcdeterministic\ VASS.
Finally in Section~\ref{sec:future} we discuss the implications of our results and sketch possible future research directions.

\section{Preliminaries}\label{sec:prelim}

\myparagraph{Basic notions}
By $\N$ and $\Z$ we denote natural and integer numbers, respectively.
For $a, b \in \N$ we write $[a,b]$ to denote the set $\{a, a+1, \ldots, b-1, b\}$.
For a vector $v \in \N^d$ and $i \in [1,d]$ we write $v[i]$ to denote the $i$-th coordinate
of vector $v$. By $0^d$ we denote the vector $v \in \N^d$ with all the coordinates equal to zero.
For a word $w = a_1 \cdot \ldots \cdot a_n$ and $1 \leq i \leq j \leq n$ we write $w[i..j] = a_i \cdot \ldots \cdot a_j$
for the infix of $w$ starting at position $i$ and ending at position $j$. We also write $w[i] = a_i$.
For any $1 \leq i \leq d$ by $e_i \in \N^d$ we denote the vector with all coordinates equal to zero except the $i$-th coordinate, which is equal to one.
For a finite alphabet $\Sigma$ we denote $\Sigma_\eps = \Sigma \cup \set{\eps}$ the extension of $\Sigma$
by the empty word $\eps$.

\myparagraph{Upward and downward-closed sets}
For two vectors $u, v \in \N^d$ we say that $u \preceq v$ if for all $i \in [1,d]$ we have $u[i] \leq v[i]$.
A set $S \subseteq \N^d$ is \emph{upward-closed} if for each $u, v \in \N^d$ it holds that
$u \in S$ and $u \preceq v$ implies $v \in S$.
Similarly, a set $S \subseteq \N^d$ is \emph{downward-closed}
if for each $u, v \in \N^d$ it holds that $u \in S$ and $v \preceq u$ implies $v \in S$.
For $u \in \N^d$ we write $u \uar = \set{v \mid u \preceq v}$ for the set of all vectors bigger than $u$ w.r.t. $\preceq$
and $u \dar = \set{v \mid v \preceq u}$ for the set of all vectors smaller than $u$ w.r.t. $\preceq$.
If an upward-closed set is of the form $u \uar$, we call it an \emph{\upatom}.
Notice that if a one-dimensional set $S \subseteq \N$ is downward-closed then either $S = \N$ or $S = [0,n]$ for some $n \in \N$.
In the first case, we write $S = \omega \dar$ and in the second case $S = n \dar$.
If a downward-closed set $D \subseteq \N^d$ is of a form $D = D_1 \times \ldots \times D_d$,
where all $D_i$ for $i \in [1,d]$ are downward-closed one dimensional sets then we call $D$ a \emph{{\downatom}}.
In the literature, sometimes \emph{up-atoms} are called principal filters, and \emph{down-atoms} are called ideals.
If $D_i = (n_i) \dar$, then we also write $D = (n_1, n_2, \ldots, n_d) \dar$. In that sense, each down-atom
is of the form $u \dar$ for $u \in (\Nomega)^d$.
Notice that a {\downatom} does not have to be of a form $u \dar$ for $u \in \N^d$, for example $D = \N^d$ is not of this form,
but $D = (\omega, \ldots, \omega) \dar$.

The following two propositions will be helpful in our considerations.

\begin{prop}[\cite{DBLP:conf/concur/CzerwinskiLMMKS18} Lemma 17, \cite{ITA_1992__26_5_449_0}, \cite{Dickson}]
\label{prop:down-decomposition}
Each downward-closed set in $\N^d$ is a finite union of {\downatom}s.
Similarly, each upward-closed set in $\N^d$ is a finite union of {\upatom}s.
\end{prop}

We represent upward-closed sets as finite unions of up-atoms and downward-closed sets as finite unions of down-atoms, the
numbers are encoded in binary.
The size of the representation of an upward- or a downward-closed set $S$ is denoted $\norm{S}$.
The following proposition helps to control the blowup of the representations of upward- and downward-closed sets.

\begin{prop}\label{prop:blowup}
Let $U \subseteq \N^d$ be an upward-closed set and $D \subseteq \N^d$ be a downward-closed set.
Then the size of the representation of their complements
$\overline{U} = \N^d \setminus U$ and $\overline{D} = \N^d \setminus D$
is at most exponential w.r.t. the sizes $\norm{U}$ and $\norm{D}$, respectively
and can be computed in exponential time.
\end{prop}

\begin{proof}
Here we present only the proof for the complement of the upward-closed set $U$
as the case for downward-closed sets follows the same lines.
Let $U = u_1\uar \cup u_2\uar \cup \ldots \cup u_n\uar$. The complement of $U$ is a downward closed set, that can be represented as a union of down-atoms. As each down-atom is represented as a downward closure of a single vector over $\N\cup\set{\omega}$ it is sufficient to calculate these vectors. Each vector $x$ representing a down-atom has two properties:
\begin{enumerate}
 \item it is in the complement of $U$ i.e. for every $i\in \set{1, \ldots, n}$ there is a coordinate $f(i)\in \set{1, \ldots, d}$ such that $x[f(i)]<u_i[f(i)]$.
 \item it is maximal, i.e. for every $y \neq x$ that satisfies the first property it holds $x\not\preceq y$.
\end{enumerate}
Intuitively, function $f$ witnesses why $x$ is not in $U$. 

Thus, for every function $f : \set{1, \ldots, n}\xrightarrow{}\set{1, \ldots, d}$ we can define one down-atom represented by
some vector $x_f$ as the maximal vector, for which function $f$ witnesses that $x_f \not\in U$.
For a coordinate $i \in [1,d]$ we need that $x_f[i] < u_j[i]$ for each $j \in [1,n]$ such that $j = f(i)$.
Therefore
$$x_f[i]= \min \set{ u_j[i]-1 \mid j\in [1,n], f(j)=i},$$
where we put $\min(\emptyset) = \omega.$
Notice that the union of $x_f\downarrow$ is exactly the complement of $U$.

Observe that there is at most exponentially many different functions $f\in \set{1\ldots n}\xrightarrow{}\set{1, \ldots, d}$ and numbers that appear in descriptions of vectors $x_f$ are $\omega$ or are smaller than
the maximal number that appear in the description of $U$. Therefore, we conclude that the proposition is true.
\end{proof}

For a more general study (for arbitrary well-quasi orders) see~\cite{Goubault-Larrecq2020}.

\myparagraph{Vector Addition Systems with States}
A $d$-dimensional Vector Addition System with States ($d$-VASS or simply VASS)
$V$ consists of a finite \emph{alphabet} $\Sigma$, a finite set of \emph{states} $Q$,
a finite set of \emph{transitions} $T \subseteq Q \times \Sigma \times \Z^d \times Q$,
a distinguished \emph{initial configuration} $c_I \in Q \times \N^d$, and a set of distinguished
\emph{final configurations} $F \subseteq Q \times \N^d$.
We write $V = (\Sigma, Q, T, c_I, F)$.
Sometimes we ignore some of the components in a VASS if they are not relevant, for example we write $V = (Q, T)$
if $\Sigma$, $c_I$, and $F$ do not matter.
A \emph{configuration} of a $d$-VASS is a pair $(q, v) \in Q \times \N^d$, we often write it $q(v)$ instead of $(q, v)$.
We write $\st(q(v)) = q$.
The set of all configurations is denoted $\conf = Q \times \N^d$.
For a state $q \in Q$ and a set $U \subseteq \N^d$ we write $q(U) = \set{q(u) \mid u \in U}$.
A transition $t = (p, a, u, q) \in T$ can be \emph{fired} in the configuration $r(v)$ if $p = r$ and $u+v \in \N^d$.
Then we write $p(v) \trans{t} q(u+v)$.
We say that the transition $t \in T$ is \emph{over} the letter $a \in \Sigma$
or the letter $a$ \emph{labels} the transition $t$.
We write $p(v) \trans{a} q(u+v)$ slightly overloading the notation,
when we want to emphasise that the transition is over the letter $a$.
The \emph{effect} of a transition $t = (p, a, u, q)$ is the vector $u$, we write $\eff(t) = u$.
The size of VASS $V$ is the total number of bits needed to represent the tuple $(\Sigma, Q, T, c_I, F)$,
we do not specify here how we represent $F$ as it may depend a lot on the form of $F$.
A sequence
$\rho = (c_1, t_1, c'_1), (c_2, t_2, c'_2), \ldots, (c_n, t_n, c'_n) \in (\conf \times T \times \conf)^{*}$
is a \emph{run} of VASS $V = (Q, T)$ if for all $i \in [1,n]$ we have $c_i \trans{t_i} c'_i$
and for all $i \in [1,n-1]$ we have $c'_i = c_{i+1}$.
We write $\tr(\rho) = t_1 \cdot \ldots \cdot t_n$.
We extend the notion of \emph{labelling} to runs;
\newcommand{\red}[1]{\textcolor{red}{#1}}
Labelling of a run $\rho$ is a word $w$ which is the concatenation of labels of transitions of $\rho$,
we also say that the run $\rho$ is \emph{over} the word $w$.
This run $\rho$,
if not empty,
is \emph{from} the configuration $c_1$ \emph{to} the configuration $c'_n$
and the configuration $c'_n$ is \emph{reachable} from the configuration $c_1$ by the run $\rho$.
We write then $c_1 \trans{\rho} c'_n$, $c_1 \trans{w} c'_n$ if $w$ labels $\rho$
slightly overloading the notation or simply $c_1 \reaches c'_n$
if the run $\rho$ is not relevant.
The empty run can go from any configuration to itself, and it is labeled with $\varepsilon$.

\myparagraph{VASS languages}
A run $\rho$ is \emph{accepting} if it is from the initial configuration to some final configuration.
For a VASS $V = (\Sigma, Q, T, c_I, F)$ we define the language of $V$
as the set of all labellings of accepting runs, namely
\[
L(V) = \set{w \in \Sigma^* \mid c_I \trans{w} c_F \text{ for some } c_F \in F}.
\]
For any configuration $c$ of $V$ we define the \emph{language of configuration $c$}, denoted $L_c(V)$
to be the language of VASS $(\Sigma, Q, T, c, F)$, namely the language of VASS $V$ with the initial
configuration $c_I$ substituted by $c$. Sometimes we simply write $L(c)$ instead of $L_c(V)$ if $V$ is clear from the context.
Further, we say that the \emph{configuration $c$ has the empty language} if $L(c)=\emptyset$.
For a VASS $V = (\Sigma, Q, T, c_I, F)$ its \emph{control automaton} is intuitively VASS $V$ after ignoring its counters.
Precisely speaking, the control automaton is $(\Sigma, Q, T', q_I, F')$ where $q_I = \st(c_I)$,
$F' = \set{q \in Q \mid \exists_{v \in \N^d} \ q(v) \in F}$ and for each $(q, a, v, q') \in T$ we have $(q, a, q') \in T'$.

Notice that a $0$-VASS, namely a VASS with no counters is just a finite automaton, so all the VASS terminology
works also for finite automata. In particular, a configuration of a $0$-VASS is simply an automaton state.
In that special case for each state $q \in Q$ we call the $L(q)$ the language of state $q$.

A VASS is \emph{deterministic} if for each configuration $c$ reachable from the initial configuration $c_I$
and for each letter $a \in \Sigma$ there is at most one configuration $c'$ such that $c \trans{a} c'$.
A VASS is \emph{$k$-ambiguous} for $k \in \N$ if for each word $w \in \Sigma^*$ there are at most $k$
accepting runs over $w$. If a VASS is $1$-ambiguous we also call it \emph{unambiguous}.

Note that, the set of languages accepted by unambiguous VASS is a strict superset of the languages accepted by deterministic VASS.
To see that unambiguous VASS can accept more, consider a language $(a^*b)^*a^n b^m$ where $n\geq m$.
On the one hand, an unambiguous VASS that accepts the language, guesses where the last block of
letter $a$ starts, then it counts the number of $a$'s in this last block, and finally counts down reading $b$'s. 
As there is exactly one correct guess, this VASS is unambiguous.
On the other hand, deterministic system can not accept the language, as intuitively speaking
it does not know whether the last block of $a$'s has already started or not.
To formulate the argument precisely one should use rather easy pumping techniques.

An example language which can be recognised by $k$-ambiguous VASS,
but probably cannot be recognised by a $(k-1)$-ambiguous VASS is the following:
language of words over $\{a,b\}$ with exactly $k$ letters $b$, which divides the word into $k+1$ blocks consisting of letters $a$, such that the first block is not the shortest one.
A $k$-ambiguous $1$-VASS increments the counter on each letter $a$ in the first block, later guesses a block which is the shortest one, and in that guessed block it decrements the counter.
As there are at most $k$ other blocks, which may be guessed, there are at most $k$ accepting runs for each word.
However, it is not clear how to accept this language by a $(k-1)$-ambiguous VASS.
The problematic scenario is if the first block is the longest one.

The following two problems are the main focus of this paper, for different subclasses of VASS:

\vspace{0.3cm}

\probl{Inclusion problem for VASS}{Two VASS $V_1$ and $V_2$}{Does $L(V_1) \subseteq L(V_2)$ hold}

\vspace{0.3cm}

\probl{Equivalence problem for VASS}{Two VASS $V_1$ and $V_2$}{Does $L(V_1) = L(V_2)$ hold}

\vspace{0.3cm}

In the sequel, we are mostly interested in VASS with the set of final configurations $F$ of some special form.
We extend the order $\preceq$ on the vectors of $\N^d$ to configurations from $Q \times \N^d$ in a natural way:
we say that $q_1(v_1) \preceq q_2(v_2)$ if $q_1 = q_2$ and $v_1 \preceq v_2$. We define the notions of upward-closed,
downward-closed, {\upatom} and {\downatom} the same as for vectors.
As Proposition~\ref{prop:down-decomposition} holds for any well quasi-order,
it also applies to $Q \times \N^d$. Proposition~\ref{prop:blowup} applies here as well,
as the upper bound on the size can be shown separately for each state.
Let the set of final configurations of VASS $V$ be $F$.
If $F$ is upward-closed then we call $V$ an \emph{upward-VASS}.
If $F$ is closed downward, then we call $V$ a \emph{downward-VASS}. 
For two sets $A \subseteq \N^a$, $B \subseteq \N^b$ and a subset of coordinates $J \subseteq [1,a+b]$ by
$A \times_J B$ we denote the set of vectors in $\N^{a+b}$ which projected into coordinates in $J$ belong to $A$
and projected into coordinates outside $J$ belong to $B$.
If $F = \bigcup_{i \in [1,n]} q_i(U_i \times_{J_i} D_i)$
where for all $i \in [1,n]$ we have $J_i \subseteq [1,d]$, $U_i \subseteq \N^{|J_i|}$ are up-atoms
and $D_i \subseteq \N^{d-|J_i|}$ are down-atoms then we call $V$ an \emph{updown-VASS}.
In the sequel we write simply $\times$ instead of $\times_J$, as the set of coordinates $J$ is never relevant.
If $F = \set{c_F}$ is a singleton, then we call $V$ a \emph{singleton-VASS}.

\noindent{\bf Notation.}
As in this paper we mostly work with upward-VASS
we often say simply a VASS instead of an upward-VASS. In other words, if not indicated otherwise, we assume that the set of
final configurations $F$ is upward-closed.

For the complexity analysis we assume that whenever $F$ is upward- or downward-closed then it is given as a union of atoms.
If $F = \bigcup_{i \in [1,n]} q_i(U_i \times D_i)$ then in the input we get a sequence of $q_i$ and representations of atoms $U_i, D_i$
defining individual sets $q_i(U_i \times D_i)$. 

\myparagraph{Language emptiness problem for VASS}
The following emptiness problem is the central problem for VASS.

\begin{quote}\label{qu:emptiness}
\textbf{Emptiness problem for VASS}
\begin{description}
  \item[Input] A VASS $V = (\Sigma, Q, T, c_I, F)$
  \item[Question] Is $c_F$ reachable from $c_I$ in $V$ for some $c_F \in F$?
\end{description}
\end{quote}

Observe that the emptiness problem is not influenced in any way by labels of the transitions,
so sometimes we will not even specify transition labels when we work with the emptiness problem.
If we want to emphasise that transition labels do not matter for some problem, then we write $V = (Q, T, c_I, F)$
ignoring the $\Sigma$ component. In such cases we also assume that transitions do not contain the $\Sigma$ component,
namely $T \subseteq Q \times \Z^d \times Q$.

Note also that the celebrated reachability problem and the coverability problem for VASS are special cases
of the emptiness problem. The reachability problem is the case when $F$ is a singleton set $\set{c_F}$,
classically it is formulated as the question whether there is a run from $c_I$ to $c_F$.
The coverability problem is the case when $F$ is an {\upatom} $c_F$,
classically it is formulated as the question whether there is a run from $c_I$ to any $c$ such that $c_F \preceq c$.
Recall that the reachability problem, so the emptiness problem for singleton-VASS is in \ackermann~\cite{DBLP:conf/lics/LerouxS19}
and actually \ackermann-complete~\cite{DBLP:conf/focs/Leroux21,DBLP:conf/focs/CzerwinskiO21}.

A special case of the emptiness problem is helpful for us in Section~\ref{sec:detvass}.

\begin{lem}\label{lem:reachability-up-down}
The emptiness problem for VASS with the acceptance condition $F = q_F(U \times D)$
where $D$ is a {\downatom} and $U$ is an {\upatom} is in \ackermann.
\end{lem}

\begin{proof}
We provide a polynomial reduction of the problem to the emptiness problem in singleton-VASS which is in \ackermann.
Let $V = (Q, T, c_I, q_F(U \times D))$ be a $d$-VASS with {\upatom} $U \subseteq \N^{d_1}$
and {\downatom} $D \subseteq \N^{d_2}$ such that $d_1+d_2 = d$.
Let $U = u \uar$ for some $u \in \N^{d_1}$ and let $D = v \dar$ for some $v \in (\Nomega)^{d_2}$.
Let us assume wlog of generality that $d_2 = d_U + d_B$ such that for $i \in [1,d_U]$ we have $v[i] = \omega$
and for $i \in [d_U + 1, d_2]$ we have $v[i] \in \N$.
Let a $d$-VASS $V'$ be the VASS $V$ slightly modified in the following way.
First we add a new state $q'_F$ and a transition $(q_F, 0^d, q'_F)$.
Next, for each coordinate
$i \in [1,d_1]$ we add a loop in state $q'_F$ (transition from $q'_F$ to $q'_F$)
with the effect $-e_i$, namely the one decreasing the
coordinate $i$, these are the coordinates corresponding to the {\upatom} $U$.
Similarly for each coordinate
$i \in [d_1+1, d_1+d_U]$ we add in $q'_F$ a loop with the effect $-e_i$,
these are the unbounded
coordinates corresponding to the {\downatom} $D$.
Finally, for each coordinate
 $i \in [d_1+d_U+1, d]$ we add in $q'_F$ a loop with the effect $e_i$ (notice
that this time we increase the counter values); these are the bounded
coordinates corresponding to {\downatom} $D$.
Let the initial configuration of $V'$ be $c_I$ (the same as in V) and the set of final configurations $F'$ of $V'$
be the singleton set containing $q'_F(u,  (0^{d_U}, v[d_U+1], \ldots, v[d_U +  d_B]))$.
Clearly $V'$ is a singleton-VASS, so the emptiness problem for $V'$ is in \ackermann.
It is easy to see that the emptiness problems in $V$ and in $V'$ are equivalent, which finishes the proof.
\end{proof}

The following is a simple and useful corollary of Lemma~\ref{lem:reachability-up-down}.

\begin{cor}\label{cor:updown}
The emptiness problem for updown-VASS is in \ackermann.
\end{cor}

\begin{proof}
Recall that for updown-VASS the acceptance condition is a finite union of $q(U \times D)$ for
some up-atom $U \subseteq \N^{d_1}$ and down-atom $D \subseteq \N^{d_2}$
where $d_1$ and $d_2$ sums to the dimension of the VASS $V$.
Thus, emptiness of the updown-VASS can be reduced to finitely many emptiness queries of the form $q(U \times D)$,
which can be decided in \ackermann due to Lemma~\ref{lem:reachability-up-down}. Notice that the number of queries
is not bigger than the size of the representation of $F$ thus the emptiness problem for updown-VASS is also in \ackermann.
\end{proof}

By Proposition~\ref{prop:down-decomposition} each downward-VASS is also an updown-VASS,
thus Corollary~\ref{cor:updown} implies the following one.

\begin{cor}\label{cor:downward}
The emptiness problem for downward-VASS is in \ackermann.
\end{cor}

Recall that the coverability problem in VASS is in \expspace~\cite{DBLP:journals/tcs/Rackoff78},
and the coverability problem is equivalent to the emptiness problem for the set of final configurations being an \upatom.
By Proposition~\ref{prop:down-decomposition} we have the following simple corollary
which creates an elegant duality for the emptiness problems in VASS.

\begin{cor}\label{cor:upward}
The emptiness problem for upward-VASS is in \expspace.
\end{cor}

Actually, even the following stronger fact is true and helpful for us in the remaining part of the paper.
The following proposition is an easy consequence of Corollary 4.6 from~\cite{DBLP:journals/iandc/LazicS21}.

\begin{prop}\label{prop:empty-language}
For each upward-VASS the representation of the downward-closed set of configurations with the empty language
can be computed in doubly-exponential time.
\end{prop}

\myparagraph{Deciding unambiguity and $k$-ambiguity for VASS}
It is a natural question to ask whether a given VASS is unambiguous, or more generally, $k$-ambiguous for a given $k \in \N$.
Deciding it can be rather easily reduced to the language emptiness problem
(and thus to the coverability or the reachability problem for VASS).
Intuitively, to check whether a VASS $V$ is unambiguous, one can construct a VASS $V'$ that accepts words having at least
two different runs in $V$. Then $V$ is unambiguous if and only if the language of $V'$ is empty. VASS $V'$ essentially speaking
simulates two copies of $V$ and additionally keeps in its state information whether these runs have already differed or not.
Thus, for VASS accepting by a set of states, the unambiguity problem can be decided in ExpSpace,
for details, see Proposition 20 in~\cite{DBLP:conf/concur/CzerwinskiFH20}.
The same complexity can be achieved for upward-VASS in exactly the same way.
Also, in the same way, for downward-VASS the unambiguity problem is reduced to the reachability problem for VASS,
which is in Ackermann~\cite{DBLP:conf/lics/LerouxS19}.

A very similar construction can be used to decide the $k$-ambiguity of the upward- or downward-VASS for any fixed $k$.
For a $d$-VASS $V$ one can construct a $d(k+1)$-VASS $V'$, which simulates $k+1$ copies of $V$,
makes sure that all of them accept and runs in all the copies are different.
Then $V$ is not $k$-ambiguous if and only if the language of $V'$ is nonempty.
In turn, the problem of checking $k$-ambiguity is in \expspace for upward-VASS
and in \ackermann for downward-VASS.

\section{Deterministic VASS}\label{sec:detvass}

\subsection{Lower bound}\label{sec:Lower_bound}
First we prove a lemma, which easily implies Theorem~\ref{thm:ackermann-hardness}.

\begin{lem}\label{lem:equality-hardness}
For each $d$-dimensional singleton-VASS $V$ with final configuration being $c_F = q_F(0^d)$
one can construct in polynomial time two deterministic $(d+1)$-dimensional upward-VASSes
$V_1$ and $V_2$ such that
\[
L(V_1) = L(V_2) \iff L(V) = \emptyset.
\]
\end{lem}

Notice that Lemma~\ref{lem:equality-hardness}
shows that the emptiness problem for a singleton-VASS with the final configuration having zero counter values
can be reduced in polynomial time to the language equivalence for deterministic VASS.
This proves Theorem~\ref{thm:ackermann-hardness} as the emptiness problem,
even with zero counter values of the final configuration
is \ackermann-hard~\cite{DBLP:conf/focs/Leroux21,DBLP:conf/focs/CzerwinskiO21}.
Also, Lemma~\ref{lem:equality-hardness} implies hardness in fixed dimensional VASS. 
For example the language equivalence problem for deterministic $(2d+5)$-VASS is $\F_d$-hard
as the emptiness problem is $\F_d$-hard for $(2d+4)$-dimensional singleton-VASS~\cite{DBLP:journals/corr/abs-2104-12695}.

\begin{proof}[Proof of Lemma~\ref{lem:equality-hardness}]
We first sketch the proof. To show the lemma, we take $V$ and add to it one transition labelled with a new letter. 
In $V_1$ the added transition can be performed if we have reached a configuration bigger than or equal to $c_F$. In $V_2$ the added transition can be performed only if we have reached a configuration strictly bigger than $c_F$. Then it is easy to see that $L(V_1)\neq L(V_2)$ if and only if $c_F$ can be reached. A detailed proof follows. 

For a given $V = (Q, T, c_I, c_F)$ we construct 
$V_1 = (\Sigma=T \cup \{a\}, Q \cup \{q'_{F}\}, T' \cup \{t_1\}, c'_I, q'_{F}(0^{d+1}\uar)),$ 
and $V_2 = (\Sigma=T \cup \{a\}, Q \cup \{q_{F}'\}, T' \cup \{t_2\}, c'_I, q'_{F}(0^{d+1}\uar))$. 
Notice that $V_1$ and $V_2$ are pretty similar to each other and also to $V$.
Both $V_1$ and $V_2$ have the same states as $V$ plus one additional state $q'_F$.
Notice that the alphabet of labels of $V_1$ and $V_2$ is the set of transitions $T$ of $V$ plus one additional letter $a$.
For each transition $t = (p, v, q) \in T$ of $V$ we create a transition $(p, t, v', q) \in T'$ where
\begin{itemize}
  \item for each $i \in [1,d]$ we have $v'[i] = v[i]$; and
  \item $v'[d+1] = v[1] + \ldots + v[d]$,
\end{itemize} 
so $v'$ is identical as $v$ on the first $d$ coordinates
and on the last $(d+1)$-th coordinate it keeps the sum of all the others.
Notice that the transitions in $T'$ are used both in $V_1$ and in $V_2$.

We also add one additional transition $t_1$ to $V_1$ and one $t_2$ to $V_2$.
To $V_1$ we add a new $a$-labelled transition from $q_F$ to $q'_F$ with the effect $0^{d+1}$
for the additional letter $a$.
To $V_2$ we also add an $a$-labelled transition between $q_F$ and $q'_F$, but with an effect equal $(0^d, -1)$.
This $-1$ on the last coordinate is the only difference between $V_1$ and $V_2$.
The starting configuration in both $V_1$ and $V_2$ is $c'_I = q_I(x_1,x_2,\ldots x_d,\sum_{i=1}^{d}x_i)$ where $c_I = q_I(x_1,x_2,\ldots x_d)$.
The set of accepting configurations is the same in both $V_1$ and $V_2$, namely it is $q'_F(0^{d+1}\uar)$ .
Notice that both $V_1$ and $V_2$ are deterministic upward-VASS, as required in the lemma statement.

Now we aim to show that $L(V_1) = L(V_2)$ if and only if $L(V) = \emptyset$.
First, observe that $L(V_1) \supseteq L(V_2)$. Clearly if $w \in L(V_2)$ then $w = ua$ for some $u \in T^*$,
where $T$ is the set of transitions of $V$. For any word $ua \in L(V_2)$, we have
\[
c'_I \trans{u} q_F(v) \trans{a} q'_F(v - e_{d+1})
\]
in $V_2$. But, then we have also
\[
c'_I \trans{u} q_F(v) \trans{a} q'_F(v)
\]
in $V_1$. Thus $ua \in L(V_1)$.

Now we show that, if $L(V) \neq \emptyset$, so $c_I \reaches q_F(0^d)$ in $V$ then $L(V_1) \neq L(V_2)$.
Let the run $\rho$ of $V$ be such that $c_I \trans{\rho} q_F(0^d)$ and let $u = \tr(\rho) \in T^*$.
Then clearly $c'_I\xrightarrow{u}q_F(0^{d+1})\xrightarrow{a}q_F'(0^{d+1})$ and $ua \in L(V_1)$. However $ua \not\in L(V_2)$ as the last coordinate on the run 
of $V_2$
over $ua$ corresponding to $\rho$ would go below zero and
this is the only possible run of $V_2$ over $ua$ due to determinism of $V_2$.

It remains to show that if $L(V) = \emptyset$, so $c_I \ntrans{} q_F(0^d)$ in $V$, then $L(V_1) \subseteq L(V_2)$.
Let $w \in L(V_1)$. Then $w = ua$ for some $u \in T^*$.
Let $c'_I \trans{\rho} c$ in $V_1$ such that $\tr(\rho) = u$. As $ua \in L(V_1)$ we know that $c = q_F(v)$.
However as $c_I \ntrans{} q_F(0^d)$ in $V$ we know that $v \neq 0^{d+1}$. In particular $v[d+1] > 0$.
Therefore $w = ua \in L(V_2)$ as the last transition over $a$ may decrease the $(d+1)$-th coordinate
and reach an accepting configuration. This finishes the proof.
\end{proof}

\subsection{Upper bounds}\label{sec:upper}
In this Section we prove three results of the form: if $V_1$ is a VASS
and $V_2$ is a VASS of some special type then deciding whether $L(V_1) \subseteq L(V_2)$
is in \ackermann. Our approach to these problems is the same, namely we first prove that
complement of $L(V_2)$ for $V_2$ of the special type is also a language of some VASS $V'_2$.
Then to decide the inclusion problem it is enough to construct VASS $V$ such that
$L(V) = L(V_1) \cap L(V'_2) = L(V_1) \setminus L(V_2)$ and check it for emptiness.
In the description above, using the term VASS we do not specify the form of its set of accepting configurations.
Starting from now on, we call upward-VASS simply VASS and for VASS with other acceptance conditions
we use their full name (like downward-VASS or updown-VASS) to distinguish them from upward-VASS.
The following lemma is very useful in our strategy of deciding the inclusion problem for VASS languages.

\begin{lem}\label{lem:language-intersection}
There is an algorithm that for a VASS $V_1$ and a downward-VASS $V_2$ constructs in polynomial
time an updown-VASS $V$ such that $L(V) = L(V_1) \cap L(V_2)$.
\end{lem}

\begin{proof}
We construct $V$ as the standard synchronous product of $V_1$ and $V_2$.
The set of accepting configurations in $V$ is also the product of accepting configurations in $V_1$
and accepting configurations in $V_2$,
thus due to Proposition~\ref{prop:down-decomposition}
a finite union of $q(U \times D)$ for a state $q$ of $V$, an up-atom $U$ and a down-atom $D$.
\end{proof}

\myparagraph{Deterministic VASS}
We first show the following theorem, which will be generalised by the other results in this section.
We aim to prove it independently in order to mildly introduce our techniques.

\begin{thm}\label{thm:complement-dvass}
There is algorithm, that for a deterministic VASS $V$ constructs, in exponential time, a downward VASS that recognises the complement
of the language of $V$.
\end{thm}

\begin{proof}[Proof of Theorem~\ref{thm:complement-dvass}]
Let $V = (\Sigma, Q, T, c_I, F)$ be a deterministic $d$-VASS.
We aim at constructing a $d$-dimensional downward-VASS $V'$ such that $L(V') = \overline{L(V)}$.
Before constructing $V'$ let us observe that there are three possible scenarios for a word $w$ to be not in $L(V)$.
The first scenario (1) is that the only run over $w$ in $V$ finishes in a non-accepting configuration.
Another possibility is that there is even no run over $w$.
Namely, for some prefix $va$ of $w$ where $v \in \Sigma^*$ and $a \in \Sigma$ we have $c_I \trans{v} c$ for some
configuration $c$ but there is no transition from $c$ over the letter $a$ as either
(2) all possible transitions over $a$ would decrease some of the counters below zero,
(3) there is no such transition possible in $V$ in the state of $c$.
For each case we separately design a part of a downward-VASS accepting it. Cases (1) and (3) are simple. For case (2) we nondeterministically guess the moment when the run would go below zero and freeze the configuration at that moment. Then, at the end of the word, we check if our guess was correct.
Notice that the set of configurations from which a step labelled with a letter $a$ would take a counter below zero is downward-closed,
so we can check the correctness of our guess using a downward-closed accepting condition. A detailed proof follows.

We are ready to describe VASS $V' = (\Sigma, Q', T', c'_I, F')$. Basically speaking, it consists of
$2 \cdot |\Sigma| + 1$ copies
of $V$. Concretely, the set of states $Q'$ is the set of pairs $Q \times (\Sigma \times\{2,3\} \cup \set{-})$.
Let $c_I = q_I(v_I)$. Then let $q'_I \in Q'$ be defined as $q'_I = (q_I, -)$
and we define the initial configuration of $V$ as $c'_I = q'_I(v_I)$.
The set of accepting configurations $F' = F_1 \cup F_2 \cup F_3$ is a union of three sets $F_i$,
each set $F_i$ for $i \in \set{1,2,3}$ is responsible for accepting words rejected by VASS $V$ due to the scenario (i) described above. We successively describe which transitions are added to $T'$
and which configurations are added to $F'$ in order to appropriately handle various scenarios. 

We first focus on words fulfilling the scenario (1).
For states of a form $(q, -)$ the VASS $V'$ is just as $V$.
Namely for each transition $(p, a, v, q) \in T$ we add 
$(p', a, v, q')$ to $T'$ where $p' = (p, -)$ and $q' = (q, -)$.
We also add to $F'$ the following set $F_1 = \set{(q,-)(v) \mid q(v) \not\in F}$.
It is easy to see that words that fulfil scenario (1) above are accepted in $V'$ by the use of the set $F_1$.
The size of the description of $F_1$ is at most exponential
compared to the size of the description of $F$ by Proposition~\ref{prop:blowup}.

Now we describe the second part of $V'$ which is responsible for
words rejected by $V$ because of the scenario (2). The idea is to guess when the run over $w$ is finished.
For each label $a \in \Sigma$ we add $(p', a, 0^d, p'')$
to $T'$
where $p' = (p, -)$ and $p'' = (p, (a,2))$. The idea is that the run reaches the configuration in which the transition labeled with $a$
cannot be fired. Now we have to check that our guess is correct. In the state $(p, (a,2))$ for $t \in T$ no transition changes the configuration.
Namely for each $p'' = (p, (a,2)) \in Q \times \Sigma$ and each $b \in \Sigma$ we add to $T'$ a transition $(p'', b, 0^d, p'')$.
We now add to $F'$ the set $F_2 = \set{(p,(a,2))(v) \mid v + \eff(t) \not\in \N^d \text{ for all } t=(p,a,\eff(t), q)\in T  \text{ where } q\in Q}$.
Notice that $F_2$ can be easily represented as a polynomial union of down-atoms.
It is easy to see that indeed $V'$ accepts by $F_2$ exactly words $w$ such that there is a run of $V$
over some prefix $v$ of $w$ but reading the next letter would decrease one of the counters below zero.
The last part of $V'$ is responsible for the words $w$ rejected by $V$ because of the scenario (3),
that is, $w$ has a prefix $va$ such that there is a run over $v \in \Sigma^*$ in $V$
but then in the state of the reached configuration there is no transition over the letter $a \in \Sigma$.
To accept such words for each state $p \in Q$ and letter $a \in \Sigma$ such that there is no transition of
a form $(p, a, v, q) \in T$ for any
$v \in \Z^d$ and $q \in Q$ we add to $T'$ transition $((p, -), a, 0^d, (p, (a,3)))$.
In each state $p' = (p, (a,3)) \in Q \times \Sigma$, we have a transition $(p', b, 0^d, p')$ for each $b \in \Sigma$.
We also add to $F'$ the set
$F_3 = \set{(p,(a,3))(v) \mid v \in \N^d \text{ and there is no } (p, a, u, q) \in T \text{ for } u \in \Z^d \text{ and } q \in Q}$.
The size of $F_3$ is polynomial w.r.t. $T$.

Summarising $V'$ with the accepting downward-closed set $F = F_1 \cup F_2 \cup F_3$ indeed satisfies $L(V') = \overline{L(V)}$,
which finishes the construction and the proof.
\end{proof}

The following theorem is a simple corollary of Theorem~\ref{thm:complement-dvass}, Lemma~\ref{lem:language-intersection}
and Corollary~\ref{cor:updown}.

\begin{thm}\label{thm:dvass-inclusion}
The inclusion problem of a VASS language in a deterministic VASS language is in \ackermann.
\end{thm}

\myparagraph{Deterministic VASS with holes}
We define here VASS with holes, which are a useful tool to obtain
our results about unambiguous VASS in Section~\ref{sec:uvass}.
A \emph{$d$-VASS with holes} (or shortly $d$-HVASS) $V$ is defined exactly as a standard VASS,
but with an additional downward-closed set $H \subseteq Q \times \N^d$ which affects the semantics of $V$.
Namely the set of configurations of $V$ is $Q\times \N^d\setminus H$. Thus each configuration on
a run of $V$ needs not only to have nonnegative counters,
but in addition to that it can not be in the set of \emph{holes} $H$.
Additionally in HVASS we allow for transitions labelled by the empty word $\eps$, in contrast to the rest of our paper.
Due to that fact in this paragraph we often work also with VASS having $\eps$-labelled transitions,
we call such VASS the $\eps$-VASS.
As an illustration of the HVASS notion, let us consider the zero-dimensional case.
In that case, the set of holes is just a subset of states.
Therefore, HVASS in dimension zero are exactly VASS in dimension zero, so finite automata.
However, for higher dimensions, the notions of HVASS and VASS differ.

We present here a few results on languages for HVASS.
First notice that for nondeterministic HVASS it is easy to construct a language equivalent $\eps$-VASS.

\begin{lem}\label{lem:hvass-to-vass}
There is an algorithm, that for a given HVASS computes, in exponential time, a language equivalent $\eps$-VASS.
\end{lem}

\begin{proof}[Proof of Lemma~\ref{lem:hvass-to-vass}]
We first sketch our solution. At first, we observe that the complement of the set of holes is an upward-closed set $U$.
The idea behind the construction is that after each step we test if the current configuration is in $U$.
We nondeterministically guess a minimal element $x_i$ of $U$ above which the current configuration is, then we subtract $x_i$ and add 
it back. If our guess was not correct, then the run is blocked. A detailed proof follows.

Let $V = (\Sigma, Q, T, q_I(v_I), F, H)$ be a $d$-HVASS with the set of holes $H$.
We aim at constructing a $d$-VASS $V' = (\Sigma, Q', T', c'_I, F')$ such that $L(V) = L(V')$.
By Proposition~\ref{prop:blowup} we can compute in exponential time an upward-closed
set of configurations $U = (Q \times \N^d) \setminus H$.
In order to translate $V$ into a $d$-VASS $V'$ intuitively we need 
to check that each configuration on the run is not in the set $H$.
To do this, we use the representation of $U$ as a finite union
$U = \bigcup_{i \in [1,k]} q_i(u_i \uar)$ for $q_i \in Q$ and $u_i \in \N^d$.
Now for each configuration $c$ on the run of $V$ the simulating VASS $V'$ needs to check that
$c$ belongs to $q_i(u_i \uar)$ for some $i \in [1,k]$.
That is why in $V'$ after every step simulating a transition of $V$ we go into a testing gadget
and after performing the test we are ready to simulate the next step.
For that purpose we define $Q' = (Q \times \set{0, 1}) \cup \set{r_1, \ldots, r_k}$.
The states in $Q \times \set{0}$ are those before the test and the states in $Q \times \set{1}$ are the ones after the test.
The states $r_1, \ldots, r_k$ are used to perform the test.
The initial configuration $c'_I$ is defined as $(q_I, 0)(v_I)$ and the set of final configurations is defined as $F' = \set{(q,1)(v) \mid q(v) \in F}$.
For each transition $(p, a, v, q)$ in $T$ we add a corresponding transition $((p, 1), a, v, (q, 0))$ to $T'$.
In each reachable configuration $(q,0)(v)$ the VASS $V'$ nondeterministically guesses for which $i \in [1,k]$
holds $q_i(u_i) \preceq q(v)$ (which guarantees that indeed $q(v) \in U$).
In order to implement it for each $q \in Q$ and each $i \in [1,k]$ such
that $q = \st(r_i)$ we add two transitions to $T'$: the one from $(q,0)$ to $r_i$ subtracting $u_i$, namely $((q,0), \eps, -u_i, r_i)$
and the one coming back and restoring the counter values, namely $(r_i, \eps, u_i, (q,1))$.
It is easy to see that $(q,0)(v) \trans{\eps} (q,1)(v)$ if and only if $q(v) \in U$, which finishes the proof.
\end{proof}

It is important to emphasise that the above construction applied to a deterministic HVASS
does not give us a deterministic VASS, so we cannot simply reuse Theorem~\ref{thm:complement-dvass}.
Thus in order to prove the decidability of the inclusion problem for HVASS
we need to generalise Theorem~\ref{thm:complement-dvass} to HVASS.

\begin{thm}\label{thm:complement-dhvass}
There is an exponential time algorithm, that
for a deterministic HVASS $V$ computes a downward-$\eps$-VASS, which recognises the complement
of the language of $V$.
\end{thm}

\begin{proof}[Proof of Theorem~\ref{thm:complement-dhvass}]
We first sketch our solution. The proof is very similar to the proof of Theorem~\ref{thm:complement-dvass}.
In case (1) we have to check if the accepting run stays above the holes, to perform it we use
the trick from Lemma~\ref{lem:hvass-to-vass}. In case (2) we freeze the counter when the run would have to drop below zero or enter the hole. The case (3) is the same as in Theorem~\ref{thm:complement-dvass}. A more detailed proof follows.

As the proof is very similar to the proof of Theorem~\ref{thm:complement-dvass} we only sketch the key differences. Let $V$ be a deterministic HVASS and let $H \subseteq Q \times \N^d$ be the set of its holes.
Let $U = (Q \times \N^d) \setminus H$, by Proposition~\ref{prop:blowup} we know that $U = \bigcup_{i \in [1,k]} q_i(u_i\uar)$
for some states $q_i \in Q$ and vectors $u_i \in \N^d$,
and additionally $\norm{U}$ is at most exponential w.r.t. the size $\norm{H}$.

The construction of $V'$ recognising the complement of $L(V)$ is almost the same as in the proof of Theorem~\ref{thm:complement-dvass},
we need to introduce only small changes. The biggest changes are in the part of $V'$ that recognises words rejected by $V$
because of scenario (1). We need to check that after each transition, the current configuration is in $U$
(so it is not in any hole from $H$). We perform it here in the same way as in the proof of Lemma~\ref{lem:hvass-to-vass}.
Namely, we guess which $q_i(u_i \uar)$ the current configuration belongs to and check it by simple VASS modifications
(for details, look at the proof of Lemma~\ref{lem:hvass-to-vass}). The size of this part of $V'$ can
have a blow-up of at most size of $U$ times, namely the size can be multiplied by some number,
which is at most exponential w.r.t. the size $\norm{H}$.

In the part recognising words rejected by $V$ due to scenario (2), we only need to adjust the accepting set $F_2$.
Indeed, we need to accept now if we are in a configuration $(p, t)(v) \in Q \times T$ such that $v + t \not\in \N^d$
or $v + t \in H$ (in contrast to only $v + t \not\in \N^d$ in the proof of Theorem~\ref{thm:complement-dvass}).
This change does not introduce any new superlinear blow-up.

Finally the part recognising words rejected by $V$ because of scenario (3) does not need adjusting at all.
It is not hard to see that the presented construction indeed accepts the complement of $L(V)$ as before.
The constructed downward-VASS $V'$ is of at most exponential size w.r.t. the size $V$ as explained above, which finishes the proof.
\end{proof}

Now, the following theorem is an easy consequence of the shown facts.
We need only to observe that proofs of Lemma~\ref{lem:language-intersection}
and Corollary~\ref{cor:updown} work as well for $\eps$-VASS.

\begin{thm}\label{thm:hvass-inclusion}
The inclusion problem of an HVASS language in a deterministic HVASS language is in \ackermann.
\end{thm}

\begin{proof}[Proof of Theorem~\ref{thm:hvass-inclusion}]
Let $V_1 = (\Sigma, Q_1, T_1, c^1_I, F_1, H_1)$ be a $d_1$-HVASS with holes $H_1 \subseteq Q_1 \times \N^{d_1}$
and let $V_2 = (\Sigma, Q_2, T_2, c^2_I, F_2, H_2)$ be a deterministic $d_2$-HVASS with holes $H_2 \subseteq Q_2 \times \N^{d_2}$.
By Lemma~\ref{lem:hvass-to-vass} an $\eps$-VASS $V'_1$ equivalent to $V_1$ can be computed in exponential time.
By Theorem~\ref{thm:complement-dhvass} a downward-$\eps$-VASS $V'_2$ can be computed in exponential time such that
$L(V'_2) = \Sigma^* \setminus L(V_2)$. It is enough to check now whether $L(V'_1) \cap L(V'_2) = \emptyset$.
By Lemma~\ref{lem:language-intersection} (extended to $\eps$-VASS) one can compute an updown-$\eps$-VASS $V$ such that
$L(V) = L(V'_1) \cap L(V'_2)$. Finally, by Corollary~\ref{cor:updown} (also extended to $\eps$-VASS)
the emptiness problem for updown-$\eps$-VASS is in \ackermann, which finishes the proof.
\end{proof}

\myparagraph{\BCdeterministic\ VASS}
We define here a generalisation of a deterministic VASS, namely a \kcdeterministic\ VASS for $k \in \N$.
Such VASS are later used as a tool for deriving results about $k$-ambiguous VASS in Section~\ref{sec:bavass}.

We say that a finite automaton $\A = (\Sigma, Q, T, q_I, F)$ is \emph{$k$-deterministic}
if for each word $w \in \Sigma^*$ there are at most $k$ \emph{maximal runs over $w$}.
We call a run $\rho$ a \emph{maximal run over $w$} if either (1) it is a run over $w$
or (2) $w = u a v$ for $u, v \in \Sigma^*$, $a \in \Sigma$ such that
the run $\rho$ is over the prefix $u$ of $w$ but there is no possible way of extending $\rho$
by any transition labelled with the letter $a \in \Sigma$. 
Let us emphasise that here we count runs in a subtle way. We do not count only the maximal
number of active runs throughout the word, but the total number of different runs that have ever been started during the word.
To illustrate the difference better let us consider an example
finite automaton $\A$ over $\Sigma = \set{a,b}$ with two states $p, q$ and with three transitions: $(p, a, p)$, $(p, a, q)$ and $(q, b, q)$.
Then $\A$ has $n+1$ maximal runs over the word $a^n$ although only two of these runs
actually survive till the end of the input word. So $\A$ is not $2$-deterministic even though for each
input word it has at most two runs.
We say that a VASS $V = (\Sigma, Q, T, c_I, F)$ is \emph{\kcdeterministic}\ if its control automaton is $k$-deterministic.

\begin{thm}\label{thm:complement-bdvass}
There is an exponential time algorithm, that
for a \kcdeterministic\ $d$-VASS $V$ constructs
a $(k \cdot d)$-dimensional downward-VASS, which recognises the complement of the language of $V$.
\end{thm}

\begin{proof}[Proof of Theorem~\ref{thm:complement-bdvass}.]
We first sketch our solution. In the algorithm the $(k \cdot d)$-dimensional downward-VASS $V'$ simulates $k$ copies of $V$
which take care of at most $k$ different maximal runs of $V$. The accepting condition $F'$ of $V'$ verifies
whether in all copies there is a reason that the simulated maximal runs do not accept.
The reasons why each individual copy does not accept are the same as in Theorem~\ref{thm:complement-dvass}.

Before starting the proof, let us remark that it would seem natural to
first build a deterministic $(k \cdot d)$-VASS equivalent to the input \kcdeterministic\ $d$-VASS
and then apply construction from the proof of Theorem~\ref{thm:complement-dvass} to recognise its complement. However, it is not clear how to construct a deterministic $(k \cdot d)$-VASS equivalent to \kcdeterministic\
$d$-VASS, thus we compute directly a VASS recognising the complement of the input VASS language.

Let $V = (\Sigma, Q, T, c_I, F)$ be a \kcdeterministic\ $d$-VASS.
We aim to construct a $(k \cdot d)$-dimensional downward-VASS $V' = (\Sigma, Q', T', c'_I, F')$
such that $L(V') = \Sigma^* \setminus L(V)$. In this proof, also, we strongly rely
on the ideas introduced in the proof of Theorem~\ref{thm:complement-dvass}.
Recall that the idea of the construction is that $V'$ simulates $k$ copies of $V$ which take
care of different maximal runs of $V$. Then the accepting condition $F'$ of $V'$
verifies whether in all copies there is a reason the simulated maximal runs do not accept.

Recall that for a run there are three scenarios in which it is not accepted: (1) it reaches the end
of the word, but the reached configuration is not accepted, (2) at some moment it tries to decrease
some counter below zero, and (3) at some moment there is no transition available over the input letter.
In the proof of Theorem~\ref{thm:complement-dvass} it was shown how a VASS can handle all three reasons.
In short words: in case (1) it simulates the run till the end of the word and then checks that the reached configuration
is not accepting and in cases (2) and (3) it guesses the moment in which there is no valid transition available
and keeps this configuration untouched (in other words freezes it) till the end of the run when it checks by the accepting condition
that the guess was correct.
We only sketch how the downward-VASS $V'$ works without stating explicitly its states and transitions.
It starts in the configuration $c'_I$ which consists of $k$ copies of $c_I$. Then it simulates the run
in all the copies in the same way until the first time when there is a choice of transition. Then we enforce
that at least one copy follows each choice, but we allow for more than one copy to follow the same choice.
In the state of $V'$ we keep the information which copies follow the same maximal
run and which have already split.
Each copy is exactly as in the proof of Theorem~\ref{thm:complement-dvass}, it realises one of the scenarios (1), (2) or (3).
As we know that $V$ is \kcdeterministic, we are sure that all the possible runs of $V$ can be simulated by $V'$
under the condition that $V'$ correctly guesses which copies should simulate which runs.
If the guesses of $V'$ are wrong and at some point it cannot send to each branch a copy, then the run of $V'$ is rejected.

A bit more concretely, the state of $V'$ keeps the following information: (I) for each of the $k$ copies in which state in $Q$ it is,
(II) the set of copies which are frozen and for each such copy a transition which caused the freezing, (III) which copies have already split and which have not (formally speaking we keep a partition of the set of copies).

We guarantee that all the branches are explored by some copy in the following way.
If in $V$ there are $m$ different transitions over some letter $a$ from a state $q$ then
in $V'$ in every state that encodes $\ell$ copies in state $q$ not all the possible $m^\ell$ options in the product are allowed.
We allow only for these options for which each of the possible $m$ transitions of $V$ are taken in one of the $\ell$ copies.
In particular if $m = \ell$ then instead of $m^m$ options in $V'$ we have exactly $m!$ options. In each of the options,
there are actually $2^m$ possibilities for freezing the copies.
This is needed because if the transition would decrease the counters below zero (so scenario (2) is realised),
the copy may not actually fire the transition, and we should have a possibility of freezing it.
If a copy gets frozen, the control state is updated accordingly.

At the end of the run over the input word $w$ VASS $V'$ checks using the acceptance condition $F'$ that indeed
all copies have simulated all the possible maximal runs and that all reject.
It is easy to see that $F'$ is a downward-closed set, since, roughly speaking, it is a product of $k$ downward-closed accepting conditions,
which finishes the proof.
\end{proof}

Theorem~\ref{thm:complement-bdvass} together with Lemma~\ref{lem:language-intersection}
and Corollary~\ref{cor:updown} easily implies (analogously as in the proof of Theorem~\ref{thm:hvass-inclusion})
the following theorem.

\begin{thm}\label{thm:bdvass-inclusion}
The inclusion problem of a VASS language in a \kcdeterministic\ VASS language is in \ackermann.
\end{thm}

\section{Unambiguous VASS}\label{sec:uvass}
In this section we aim to prove Theorem~\ref{thm:unambiguous}.
However, possibly a more valuable contribution of this section is a novel
technique which we introduce in order to show Theorem~\ref{thm:unambiguous}.
The essence of this technique is to introduce a regular lookahead to words,
namely to decorate each letter of a word with a piece of information regarding some regular properties
of the suffix of this word. For technical reasons, we realise it by using finite monoids.

The high-level intuition behind the proof of Theorem~\ref{thm:unambiguous} is the following.
We first introduce the notion of $(M,h)$-decoration of words, languages and VASS,
where $M$ is a monoid and $h: \Sigma^* \to M$ is a homomorphism.
Proposition~\ref{prop:decoration} states that language inclusion of two VASS can
be reduced to language inclusion of its decorations.
On the other hand Theorem~\ref{thm:uvass-to-hvass} shows that for appropriately
chosen pair $(M,h)$ the decorations of unambiguous VASS are deterministic HVASS.
Theorem~\ref{thm:reg-sep-upgrade} states that such an appropriate pair can be computed quickly enough.
Thus, the language inclusion of unambiguous VASS reduces to language inclusion of deterministic HVASS,
which is in \ackermann due to Theorem~\ref{thm:hvass-inclusion}.

Recall that a monoid $M$ together with a homomorphism $h: \Sigma^* \to M$
and an accepting subset $F \subseteq M$ recognises a language $L$ if
$L = h^{-1}(F)$. In other words $L$ is exactly the set of words $w$ such that $h(w) \in F$.
The following proposition is folklore, for details see~\cite{DBLP:reference/hfl/Pin97} (Proposition 3.12).

\begin{prop}\label{prop:monoids}
A language of finite words is regular if and only if it is recognised by some finite monoid.
\end{prop}

For that reason, monoids are a good tool for working with regular languages.
In particular Proposition~\ref{prop:monoids} implies that for each finite family of regular languages
there is a monoid, which recognises all of them, this fact is useful in Theorem~\ref{thm:uvass-to-hvass}.
Fix a finite monoid $M$ and a homomorphism $h: \Sigma^* \to M$.
For a word $w = a_1 \cdot \ldots \cdot a_n \in \Sigma^*$ we define its \emph{$(M,h)$-decoration}
to be the following word over an alphabet $\Sigma_\eps \times M$:
\[
(\eps, h(a_1 \cdot \ldots \cdot a_n)) \cdot (a_1, h(a_2 \cdot \ldots \cdot a_n)) \cdot \ldots \cdot (a_{n-1}, h(a_n)) \cdot (a_n, h(\eps)).
\]
In other words, the $(M,h)$-decoration of a word $w$ of length $n$ has length $n+1$,
where the $i$-th letter has the form $(a_{i-1}, h(a_i \cdot \ldots \cdot a_n))$.
We denote the $(M,h)$-decoration of a word $w$ as $w_{(M,h)}$.
If $h(w) = m$ then we say that the word $w$ has \emph{type} $m \in M$.
The intuition behind the $(M,h)$-decoration of $w$ is that for each language $L$ which
is recognised by the pair $(M,h)$ the $i$-th letter of $w$ is extended with an information whether the suffix
of $w$ after this letter belongs to $L$ or does not belong. This information can be extracted
from the monoid element $h(a_{i+1} \cdot \ldots \cdot a_n)$ by which the letter $a_i$ is extended.
As an illustration consider words over alphabet $\Sigma = \{a, b\}$,
monoid $M = \Z_2$ counting modulo two and homomorphism $h: \Sigma \to M$ defined
as $h(a) = 1$, $h(b) = 0$. In that case for each $w \in \Sigma^*$ the element $h(w)$ indicates whether
the number of letters $a$ in the word $w$ is odd or even. The decoration of $w = a a b a b$
is then $w_{(M,h)} = (\eps, 1) (a, 0) (a, 1) (b, 1) (a, 0) (b, 0)$.

We say that a word $u \in (\Sigma_\eps \times M)^*$ is \emph{well-formed}
if $u = (\eps, m_0) \cdot (a_1, m_1) \cdot \ldots \cdot (a_n, m_n)$ such that all $a_i \in \Sigma$,
and for each $i \in [0,n]$ the type of $a_{i+1} \cdot \ldots \cdot a_n$ is $m_i$
(in particular type of $\eps$ is $m_n$).
We say that such a word $u$ \emph{projects} into word $a_1 \cdot \ldots \cdot a_n$.
It is easy to observe that $w_{(M,h)}$ is the only well-formed word that projects into $w$.
The following proposition is useful in Section~\ref{sec:bavass}, an appropriate finite automaton
can be easily constructed.

\begin{prop}\label{prop:well-formed}
The set of all well-formed words is regular.
\end{prop}

A word is \emph{almost well-formed} if it satisfies all the conditions of well-formedness, but the first letter
is not necessarily of the form $(\eps, m)$ for $m \in M$, it can as well belong to $\Sigma \times M$.

The $(M,h)$-decoration of a language $L$, denoted $L_{(M,h)}$,
is the set of all $(M,h)$-decorations of all words in $L$, namely
\[
L_{(M,h)} = \set{w_{(M,h)} \mid w \in L}.
\] 
Since the $(M,h)$-decoration is a function from the set of words over $\Sigma$ to words over
$\Sigma_\eps \times M$ we observe that $u = v$ iff $u_{(M,h)} = v_{(M,h)}$ and clearly the following proposition holds.

\begin{prop}\label{prop:decoration}
For each finite alphabet $\Sigma$, two languages $K, L \subseteq \Sigma^*$, a
monoid $M$ and homomorphism $h: \Sigma^* \to M$ we have
\[
K \subseteq L \iff K_{(M,h)} \subseteq L_{(M,h)}.
\]
\end{prop}

\myparagraph{HVASS construction}
As defined earlier, an HVASS (VASS with holes) is a VASS with some downward-closed set $H$ of prohibited configurations
(see Section~\ref{sec:detvass}, paragraph Deterministic VASS with holes).
For each $d$-VASS $V = (\Sigma, Q, T, c_I, F)$, a monoid $M$ and a homomorphism $h: \Sigma^* \to M$ we can define in a natural
way a $d$-HVASS $V_{(M,h)} = (\Sigma_\eps \times M, Q', T', c'_I, F')$ accepting the $(M,h)$-decoration of $L(V)$.
The set of states $Q'$ equals $Q \times (M \cup \set{\perp})$.
The intuition is that $V_{(M,h)}$ is designed in such a way that for any state $(q,m) \in Q \times M$
and vector $v\in \N^d$ if $(q,m)(v)\xrightarrow{w'} F'$ 
then $w'$ is almost well-formed and $w'$ projects into some $w \in \Sigma^*$ such that $h(w)=m$.
If $c_I = q_I(v_I)$ then configuration $c'_I=(q_I, \perp)(v_I)$ is the initial configuration of $V_{(M,h)}$.
The set of final configurations $F'$ is defined as $F' = \set{(q, h(\eps))(v) \mid q(v) \in F}$.
Finally we define the set of transitions $T'$ of $V'$ as follows.
First, for each $m \in M$ we add the following transition $((q_I, \perp), (\eps, m), 0^d, (q_I, m))$ to $T'$.
Then for each transition $(p, a, v, q) \in T$ and for each $m \in M$ we add to $T'$ the transition
$(p', a', v, q')$ where $a' = (a, m)$, $q' = (q, m)$ and $p' = (p, h(a) \cdot m)$.
It is now easy to see that for any word $w = a_1 \cdot \ldots \cdot a_n \in \Sigma^*$ we have
\[
q_I(v_I) \trans{a_1} q_1(v_1) \trans{a_2} \ldots \trans{a_{n-1}} q_{n-1}(v_{n-1}) \trans{a_n} q_n(v_n)
\]
if and only if
\begin{align*}
(q_I, \perp)(v_I) & \trans{(\eps, m_1)} (q_I, m_1)(v_I) \trans{(a_1, m_2)} (q_1, m_2)(v_1) \trans{(a_2, m_3)} \ldots \\
& \trans{(a_{n-1}, m_n)} (q_{n-1},m_n)(v_{n-1}) \trans{(a_n, m_{n+1})} (q_n, m_{n+1})(v_n),
\end{align*}
where $m_i = h(w[i..n])$ for all $i \in [1,n+1]$, in particular $m_{n+1} = h(\eps)$.
Therefore, indeed $L(V_{(M,h)}) = L(V)_{(M,h)}$.
Until now the defined HVASS is actually a VASS, we have not defined any holes. Our aim is now to remove configurations
with the empty language, that is, $(q,m)(v)$ for which there is no word $w \in (\Sigma_\eps \times M)^*$ such that
$(q,m)(v) \trans{w} c'_F$ for some $c'_F \in F'$. Notice that as $F'$ is upward-closed we know that the set of
configurations with the empty language is downward-closed.
This is how we define the set of holes $H$, it is exactly the set of configurations with
the empty language. We can compute the set of holes in doubly-exponential time by Proposition~\ref{prop:empty-language}.

By Proposition~\ref{prop:decoration} we know that for two VASS $U, V$ we have $L(U) \subseteq L(V)$ if and only if
$L(U_{(M,h)}) \subseteq L(V_{(M,h)})$. This equivalence is useful, as we show in a moment that for an unambiguous VASS $V$ and suitably chosen $(M,h)$
the HVASS $V_{(M,h)}$ is deterministic.

\myparagraph{Regular separability}
We use here the notion of regular separability. We say that two languages $K, L \subseteq \Sigma^*$
are \emph{regular-separable} if there exists a regular language $S \subseteq \Sigma^*$
such that $K \subseteq S$ and $S \cap L = \emptyset$. We then say that $S$ \emph{separates} $K$ and $L$
and $S$ is a \emph{separator} of $K$ and $L$. We recall here a theorem about regular-separability of VASS languages
(importantly upward-VASS languages, not downward-VASS languages) from~\cite{DBLP:conf/concur/CzerwinskiLMMKS18}.

\begin{thmC}[{\cite[Theorem 24]{DBLP:conf/concur/CzerwinskiLMMKS18}}]\label{thm:reg-sep}
For any two VASS languages $L_1, L_2 \subseteq \Sigma^*$
if $L_1 \cap L_2 = \emptyset$ then $L_1$ and $L_2$ are regular-separable
and one can compute a regular separator in elementary time.
\end{thmC}

\begin{proof}
Theorem 24 in~\cite{DBLP:conf/concur/CzerwinskiLMMKS18} says that there exists a regular separator
of $L_1$ and $L_2$ of size at most triply-exponential. In order to compute it, we can simply enumerate
all the possible separators of at most triply-exponential size and check them one by one.
For a given regular language and a given VASS language by Proposition~\ref{cor:upward}
one can check in doubly-exponential time whether they intersection is nonempty.
\end{proof}

For our purposes we need a bit stronger version of this theorem.
We say that a family of regular languages \emph{$\F$ separates languages of a VASS $V$} if for any two configurations $c_1, c_2$ such that 
languages $L(c_1)$ and $L(c_2)$ are disjoint there exists a language $S \in \F$ that separates $L(c_1)$ and $L(c_2)$.

\begin{thm}\label{thm:reg-sep-upgrade}
There is an algorithm that,
for any VASS computes, in an elementary time, a finite family of regular languages that separates its languages.
\end{thm}

\begin{proof}
Let us fix a $d$-VASS $V = (\Sigma, Q, T, c_I, F)$.
Let us define the set of pairs of configurations of $V$ with disjoint languages
$D = \set{(c_1, c_2) \mid L(c_1) \cap L(c_2) = \emptyset} \subseteq Q \times \N^d \times Q \times \N^d$.
One can easily see that the set $D$ is exactly the set of configurations with empty language in the synchronised
product of VASS $V$ with itself. Thus, by Proposition~\ref{prop:empty-language}
we can compute in doubly-exponential time its representation
as a finite union of down-atoms $D = A_1 \cup \ldots \cup A_n$.
We show now that for each $i \in [1,n]$ one can compute in elementary time a regular language $S_i$ such
that for all $(c_1, c_2) \in A_i$ the language $S_i$ separates $L(c_1)$ and $L(c_2)$.
This will complete the proof showing that one of $S_1, \ldots, S_n$ separates $L(c_1)$ and $L(c_2)$
whenever they are disjoint.

Let $A \subseteq Q \times \N^d \times Q \times \N^d$ be a down-atom.
Therefore $A = D_1 \times D_2$
where $D_1 = p_1(u_1 \dar)$ and $D_2 = p_2(u_2 \dar)$ for some $u_1, u_2 \in (\Nomega)^d$.
Let $L_1 = \bigcup_{c \in D_1} L(c)$ and $L_2 = \bigcup_{c \in D_2} L(c)$.
Languages $L_1$ and $L_2$ are disjoint as $w \in L_1 \cap L_2$ would imply $w \in L(c_1) \cap L(c_2)$
for some $c_1 \in D_1$ and $c_2 \in D_2$.
Now, observe that $L_1$ is not only an infinite union of VASS languages but also a VASS language itself.
Indeed, let $V_1=(\Sigma, Q, T_1, c_{(I,1)}, F_1)$ be the VASS $V$ where all coordinates $i \in [1,d]$ such that $u_1[i] = \omega$ are ignored.
Concretely, 
\begin{itemize}
 \item $(p,a,v_1,q)\in T_1$ if there exists $(p,a,v,q)\in T$ such that for every $i$ holds either $v_1[i]=v[i]$ or $v_1[i]=0$ and $u_1[i]=\omega$,\
 \item $\st(c_{(I,1)}) = p_1$ and for every $i$ holds either $c_{(I,1)}[i] = u_1[i]$ or $c_{(I,1)}[i] = 0$ and $u_1[i] = \omega$,
 \item $(q,v_1)\in F_1$ if there exists $(q,v)\in F$ such that for every $i$ holds either $v_1[i]=v[i]$ or $u_1[i]=\omega$.
\end{itemize}
Then it is easy to observe that $V_1$ accepts exactly the language $L_1$.
Similarly, one can define VASS $V_2$ accepting the language $L_2$.
By Theorem~\ref{thm:reg-sep} we can compute in elementary time some regular separator $S$ of $L(V_1)$ and $L(V_2)$.
It is now easy to see that for any configurations $c_1 \in D_1$ and $c_2 \in D_2$ languages $L(c_1)$ and $L(c_2)$
are separated by $S$.
\end{proof}

Now we are ready to use the notion of $(M,h)$-decoration of a VASS language. Let us recall that a regular language $L$ is recognised by
a monoid $M$ and homomorphism $h:\Sigma^*\xrightarrow{} M$ if there is $F \subseteq M$ such that $L=h^{-1}(F).$

\begin{thm}\label{thm:uvass-to-hvass}
Let $V$ be an unambiguous VASS over $\Sigma$ and $\F$ be a finite family of regular languages separating languages of $V$.
Suppose $M$ is a monoid with homomorphism $h: \Sigma^* \to M$ recognising every language in $\F$.
Then the HVASS $V_{(M,h)}$ is deterministic.
\end{thm}

\begin{proof}
Let $V = (\Sigma, Q, T, c_I, F)$ and let $c_I = q_I(v_I)$.
We aim to show that HVASS $V_{(M,h)} = (\Sigma', Q', T', c'_I, F')$ is deterministic,
where $\Sigma' = \Sigma_\eps \times M$ and $Q' = Q \times (M \cup \set{\perp})$.
It is easy to see from the definition of $V_{(M,h)}$ that for each $(a,m) \in \Sigma'$ and each $q \in Q$ the state $(q, \perp)$
has at most one outgoing transition over $(a,m)$. Indeed, there is exactly one transition over $(\eps, m)$
outgoing from $(q_I, \perp)$ and no outgoing transitions in the other cases.
Assume now towards a contradiction that $V_{(M,h)}$ is not deterministic.
Then there is some configuration $c = (q, m)(v)$ with $(q, m) \in Q \times M$ such that $c_I \trans{u} c$
for some word $u$ over $\Sigma'$ and a letter $(a, m') \in \Sigma'$ such that a transition from $c$ over $(a, m')$
leads to two different configurations $c_1 = (q_1, m')(v_1)$ and $c_2 = (q_2, m')(v_2)$.
Recall that a transition over $(a, m')$ has to lead to some state where the second component is equal to $m'$.
As configurations with empty language are not present in $V_{(M,h)}$ we know that there exist
words $w_1 \in L(c_1)$ and $w_2 \in L(c_2)$. Recall that as $c_1 = (q_1, m')(v_1)$ and $c_2 = (q_2, m')(v_2)$
we have $h(w_1) = m' = h(w_2)$. We show now that $L(c_1)$ and $L(c_2)$ are disjoint.
Assume otherwise that there exists $w \in L(c_1) \cap L(c_2)$.
Then there are at least two accepting runs over the word $u \cdot (a, m') \cdot w$ in $V_{(M,h)}$.
However, this means that there are at least two accepting runs over
the projection of $u \cdot (a, m') \cdot w$ in $V$, which contradicts the unambiguity of $V$.
Thus, $L(c_1)$ and $L(c_2)$ are disjoint and therefore separable by some language from $\F$.
Recall that all languages in $\F$ are recognisable by $(M,h)$, therefore words from $L(c_1)$ should be mapped
by homomorphism $h$ to different elements from $M$ than words from $L(c_2)$.
However, $h(w_1) = m'$ for $w_1 \in L(c_1)$ and $h(w_2) = m'$ for $w_2 \in L(c_2)$ which leads to the contradiction.
\end{proof}

Now we are ready to prove Theorem~\ref{thm:unambiguous}.
Let $V_1$ be a VASS and $V_2$ be an unambiguous VASS, both with labels from $\Sigma$.
We first compute a finite family $\F$ separating languages of $V_2$ which can be
performed in elementary time by Theorem~\ref{thm:reg-sep-upgrade}
and then we compute a finite monoid $M$ together with a homomorphism $h: \Sigma^* \to M$
recognising all the languages from $\F$. 
By Proposition~\ref{prop:decoration} we get that $L(V_1) \subseteq L(V_2)$
if and only if $L(V_1)_{(M,h)} \subseteq L(V_2)_{(M,h)}$.
We now compute HVASS $V'_1 = V_{1_{(M,h)}}$ and $V'_2 = V_{2_{(M,h)}}$ as described above in the paragraph HVASS construction.
By Theorem~\ref{thm:uvass-to-hvass} the HVASS $V'_2$ is deterministic.
Thus it remains to check whether the language of a HVASS $V'_1$ is included in the language of a deterministic HVASS $V'_2$,
which is in \ackermann due to Theorem~\ref{thm:hvass-inclusion}.

\begin{rem}
We remark that our technique can be applied not only to VASS but also
in a more general setting of well-structured transition systems.
In~\cite{DBLP:conf/concur/CzerwinskiLMMKS18} it was shown that for any well-structured transition systems
satisfying some mild conditions (finite branching is enough), the disjointness of two languages implies regular separability
of these languages. We claim that an analogue of our Theorem~\ref{thm:reg-sep-upgrade} can be obtained in that case as well.
Assume now that $\mathcal{W}_1, \mathcal{W}_2$ are two classes of finitely branching well-structured transition systems,
such that for any two systems $V_1 \in \mathcal{W}_1$, $V_2 \in \mathcal{W}_2$ where $V_2$ is deterministic
the language inclusion problem is decidable.
Then this problem is also likely to be decidable if we weaken the condition of determinism to unambiguity. More concretely
speaking this seems to be the case if it is possible to perform the construction analogous to Theorem~\ref{thm:complement-dvass}
in $\mathcal{W}_2$, namely if one can compute the system recognising the complement of deterministic language without leaving
the class $\mathcal{W}_2$. We claim that an example of such a class $\mathcal{W}_2$ is the class of VASS with one reset.
The emptiness problem for VASS with one zero-test (and thus also for VASS with one reset) is decidable
due to~\cite{DBLP:journals/entcs/Reinhardt08,DBLP:conf/mfcs/Bonnet11}. Then following our techniques it seems
that one can show that inclusion of a VASS language in a language of an unambiguous VASS with one reset is decidable.
\end{rem}

\section{Boundedly-ambiguous VASS}\label{sec:bavass}
In this section we aim to prove Theorem~\ref{thm:complement-bavass} and
show that Theorem~\ref{thm:k-ambiguous} is its easy consequence.
Recall the statements of both Theorem~\ref{thm:complement-bavass} and
Theorem~\ref{thm:k-ambiguous}.

\vspace{0.3cm}

\noindent
\textbf{Theorem~\ref{thm:complement-bavass}.}
There is an algorithm that, for each $k \in \N$ and a $k$-ambiguous VASS $V$ constructs, in elementary time,
a VASS with a downward-closed set of accepting configurations which recognises the complement of the language of $V$.

\vspace{0.3cm}

\noindent
\textbf{Theorem~\ref{thm:k-ambiguous}.}
For each $k \in \N$ the language inclusion problem of a VASS in a $k$-ambiguous VASS is in \ackermann.

\vspace{0.3cm}

Let us first show how Theorem~\ref{thm:complement-bavass} implies Theorem~\ref{thm:k-ambiguous}.
Let $V_1$ be a VASS and $V_2$ be a $k$-ambiguous VASS.
By Theorem~\ref{thm:complement-bavass} can be computed in elementary time
a downward-VASS $V'_2$ such that $L(V'_2) = \Sigma^* \setminus L(V_2)$.
By Lemma~\ref{lem:language-intersection} one can construct in time a polynomial w.r.t. the size
of $V_1$ and $V'_2$ an updown-VASS $V$ such that
$L(V) = L(V_1) \cap L(V'_2) = L(V_1) \setminus L(V_2)$.
By Corollary~\ref{cor:updown} emptiness of $V$ is decidable in \ackermann
which, consequently, proves Theorem~\ref{thm:k-ambiguous}.

The rest of this section focusses on the proof of Theorem~\ref{thm:complement-bavass}.

\begin{proof}[Proof of Theorem~\ref{thm:complement-bavass}]
We now prove Theorem~\ref{thm:complement-bavass} using Lemmas~\ref{lem:ba-control}~and~\ref{lem:kafa-to-kdfa}.
Then in Sections~\ref{sec:ba-control}~and~\ref{sec:kafa-to-kdfa} we show the formulated lemmas.
Let $V$ be a $k$-ambiguous VASS over an alphabet $\Sigma$.
First due to Lemma~\ref{lem:ba-control} proved in Section~\ref{sec:ba-control}
we construct a VASS $V^1$ which is language equivalent to $V$
and additionally has the control automaton being $k$-ambiguous.

\begin{lem}\label{lem:ba-control}
There is an algorithm, that
for each $k$-ambiguous VASS $V$ constructs, in doubly-exponential time, a language equivalent VASS $V'$
with the property that its control automaton is $k$-ambiguous.
\end{lem}

Now our aim is to obtain a \kcdeterministic\ VASS $V^2$, which is language equivalent to $V^1$.
We are not able to achieve it literally, but using the notion of $(M,h)$-decoration from Section~\ref{sec:uvass}
we can compute a somehow connected \kcdeterministic\ VASS $V^2$.
We use the following lemma which is proved in Section~\ref{sec:kafa-to-kdfa}.

\begin{lem}\label{lem:kafa-to-kdfa}
Let $\A = (\Sigma, Q, T, q, F)$ be a $k$-ambiguous finite automaton for some $k \in \N$.
Let $M$ be a finite monoid and $h: \Sigma^* \to M$ be a homomorphism
recognising all the state languages of the automaton $\A$.
Then the decoration $\A_{(M,h)}$ is a $k$-deterministic finite automaton.
\end{lem}

Now we consider the control automaton $\A$ of VASS $V^1$.
We compute a monoid $M$ together with a homomorphism $h: \Sigma^* \to M$
that recognises all the state languages of $\A$.
Then we construct the automaton $\A_{(M,h)}$. Note that the decoration of a VASS produces an HVASS, but as we decorate an automaton i.e. $0$-VASS we get a
$0$-HVASS which is also a finite automaton. 
Based on $\A_{(M,h)}$ we construct a VASS $V^2$.
We add a vector to every transition in $\A_{(M,h)}$, to produce a
VASS, that recognises the $(M,h)$-decoration of the language of VASS $V^1$.
Precisely, if we have a transition $((p,m),(a,m'),(q,m'))$ in $\A_{(M,h)}$ then it is created from the transition $(p,a,q)$ in $\A$,
which originates from the transition $(p,a,v,q)$ in $V^1$. So in $V^2$ we label $((p,m),(a,m'),(q,m'))$
with $v$ i.e. we have the transition $((p,m),(a,m'), v,(q,m'))$. 
Similarly, based on $V^1$, we define initial and final configurations in $V^2$.
It is easy to see that there is a bijection between accepting runs in $V^1$ and accepting runs in $V^2$. 
By Lemma~\ref{lem:kafa-to-kdfa} $\A_{(M,h)}$ is $k$-deterministic which immediately implies that $V^2$ is \kcdeterministic\ as well.

Now by Theorem~\ref{thm:complement-bdvass} we compute a downward-VASS $V^3$ which recognises
the complement of $L(V^2)$. Notice that for each $w \in \Sigma^*$ there is exactly one well-formed word
in $\Sigma_\eps \times M $ which projects into $w$, namely the $(M,h)$-decoration of $w$.
Therefore, $V^3$ accepts all the non-well-formed words and all the well-formed words that
project into the complement of $L(V)$. By Proposition~\ref{prop:well-formed} the set of all well-formed words is recognised
by some finite automaton $\B$. Computing a synchronised product of $\B$ and $V^3$
one can obtain a downward-VASS $V^4$ which recognises the intersection of languages $L(\B)$ and $L(V^3)$,
namely all the well-formed words which project into the complement of $L(V)$.
Now, it is easy to compute a downward-$\eps$-VASS $V^5$ recognising the projection of $L(V^4)$ into the first component
of the alphabet $\Sigma_\eps \times M$.
We obtain $V^5$ simply by ignoring the second component of the alphabet.
Thus, $V^5$ recognises exactly the complement of $L(V)$.
However $V^5$ is not a downward-VASS as it contains a few $\eps$-labelled transitions leaving the initial state.
We aim to eliminate these $\eps$-labelled transitions.
Recall that in the construction of the $(M,h)$-decoration the $(\eps, m)$-labelled transitions leaving
the initial configuration have the effect $0^d$. Thus, it is easy to eliminate them and obtain a downward-VASS $V^6$ which recognises
exactly the complement of $L(V)$, which finishes the proof of Theorem~\ref{thm:complement-bavass}.
Let us remark here that even ignoring the last step of elimination and obtaining a downward-$\eps$-VASS recognising
the complement of $L(V)$ would be enough to prove Theorem~\ref{thm:k-ambiguous} along the same lines
as it is proved now.
\end{proof}

\subsection{Proof of Lemma~\ref{lem:ba-control}}\label{sec:ba-control}

We recall the statement of Lemma~\ref{lem:ba-control}.

\vspace{0.3cm}

\noindent
\textbf{Lemma~\ref{lem:ba-control}.}
There is an algorithm, that
for each $k$-ambiguous VASS $V$ constructs, in doubly-exponential time, a language equivalent VASS $V'$
with the property that its control automaton is $k$-ambiguous.

\begin{proof}[Proof of Lemma~\ref{lem:ba-control}.] Let $V = (\Sigma, Q, T, c_I, F)$ be a $k$-ambiguous $d$-VASS for some $k \in \N$.
We aim at constructing a language equivalent VASS $V'$ such that
its control automaton is $k$-ambiguous. Notice that if the control automaton of $V'$
is $k$-ambiguous then clearly $V'$ is $k$-ambiguous as well.
The idea behind the construction of $V'$ is that it behaves as $V$,
but additionally the
states of $V'$ keep some finite information about the values of the counters.
More precisely counter values are kept exactly until they pass some threshold $M$.
After passing this threshold, its value is not kept in the state, only the information about passing
the threshold is remembered. We define the described notion below more precisely
as the $M$-abstraction of a VASS. The control state of $V'$ is the $M$-abstraction of $V$
for appropriately chosen $M$; we show how to choose $M$ later.

\begin{defi}
We consider here vectors over $\Nomega$ where $\omega$ is interpreted
as a number greater than all natural numbers and fulfilling $\omega + a = \omega$
for any $a \in \Z$. For $v \in (\Nomega)^d$ and $M \in \N$ we define $v_M \in ([0, M-1] \cup \set{\omega})^d$ such
that for all $i \in [1,d]$ we have $v_M[i] = v[i]$ if $v[i] < M$ and $v_M[i] = \omega$ otherwise.
In other words, all the numbers in $v$ that are at least equal to $M$ are changed into $\omega$ in $v_M$.
Let $V = (\Sigma, Q, T, c_I, F)$ be a $d$-VASS.
For $M \in \N$ the \emph{$M$-abstraction of $V$} is a finite automaton $V_M = (\Sigma, Q', T', q'_I, Q'_F)$
which roughly speaking behaves like $V$ up to the threshold $M$.
More concretely, we define $V_M$ as follows:
\begin{itemize}
  \item the set $Q'$ of states equals $Q \times ([0,M-1] \cup \set{\omega})^d$
  \item if $c_I = q_I(v)$ then the initial state equals $q'_I = q_I(v_M)$
  \item the set $Q'_F$ of accepting states equals $\set{q(v_M) \mid q(v) \in F}$
  \item for each transition $t = (q, a, v, q') \in T$ and for each $u \in ([0,M-1] \cup \set{\omega})^d$
  such that $u + v \in (\Nomega)^d$ we define transition $(q(u), a, q'(u')) \in T'$
  such that $u'= (u + v)_M$.
\end{itemize}
\end{defi}

Now we aim at finding $M \in \N$ such that the $M$-abstraction of $V$ is $k$-ambiguous.
Notice that it proves Lemma~\ref{lem:ba-control}.
Having such an $M$ we substitute the control automaton of $V$ by the $M$-abstraction of $V$
and obtain a VASS $V'$ that meets the conditions of Lemma~\ref{lem:ba-control}.
In other words, each state $q$ of the control automaton of $V$ multiplies itself and now
in the control automaton $V_M$ of $V'$ there are $(M+1)^d$ copies of $q$.
The effects of transitions in $V'$ are inherited from the effects of transitions in $V$.
The languages of $V$ and $V'$ are the same as the new control automaton $V_M$ of $V'$
never eliminates any run allowed by the old control automaton in $V$.

We define now a $d(k+1)$-VASS $\overline{V}$ which roughly speaking simulates $k+1$ copies of $V$
and accepts if all the copies have taken different runs and additionally all of them accept.
Notice that it is easy to construct $\overline{V}$: it is just a synchronised product of $k+1$ copies of $V$
which additionally keeps in the state the information which copies follow the same runs
and which have already split. The language of $\overline{V}$ contains those words that have at least $k+1$ runs
in $V$. As $V$ is $k$-ambiguous, we know that $L(\overline{V})$ is empty.
Now, we formulate the following lemma.

\begin{lem}\label{lem:threshold}
For each $d$-VASS $V$ there is a computable doubly-exponential threshold $M \in \N$
such that if the $M$-abstraction of $V$ has an accepting run then $V$ has an accepting run.
\end{lem}

Lemma~\ref{lem:threshold} is a consequence of the result by Rackoff~\cite{DBLP:journals/tcs/Rackoff78} showing
that if there is a covering path in VASS then there is also a covering path of at most doubly-exponential length.

By Lemma~\ref{lem:threshold} we can compute a threshold $\overline{M}$ such that the $\overline{M}$-abstraction of $\overline{V}$
has no accepting run. We claim now that the $\overline{M}$-abstraction of $V$, namely $V_{\overline{M}}$ is $k$-ambiguous.
Assume otherwise; let $V_{\overline{M}}$ have at least $k+1$ runs over some word $w$.
Then it is easy to see that the $\overline{M}$-abstraction of $\overline{V}$ has an accepting run over $w$,
which is a contradiction. Thus indeed $V_{\overline{M}}$ is $k$-ambiguous and 
extending the control automaton of $V$ by $M$-bounded counters in a way as in $V_{\overline{M}}$
(such that the control automaton becomes exactly $V_{\overline{M}}$) finishes the proof of Lemma~\ref{lem:ba-control}.
\end{proof}

\begin{proof}[Proof of Lemma~\ref{lem:threshold}]
Let $V=(\Sigma, Q, T, c_I, F)$ and let $\Set{C}$ be the set of configurations of $V$ from which it is possible to reach $F$. As the set $F$ is upward closed the same holds for $\Set{C}$. Further, let $\bar{\Set{C}}$ be $(\N\cup\omega)^d \setminus \Set{C}$,
so the complement of $\Set{C}$ in the space of $(\N\cup\omega)^d$. It is downward closed i.e. it is a union of down-atoms, $d_1 \downarrow, d_2 \downarrow, \ldots, d_n \downarrow$.
Let $\norm{V}$ be the norm of VASS $V$, namely the maximal absolute value occurring on transition of $V$.
Let $M = \max \set{d_i[j]: i\leq n, j\leq d, d_i[j]\neq \omega}+\norm{V}$ be the maximal constant that appears in the description of the set $\bar{\Set{C}}$ plus the norm of VASS $V$.
By~\cite{DBLP:journals/tcs/Rackoff78} the minimal length of a path covering a given configuration is at most doubly-exponential wrt. the
size of the VASS and the norm of the target configuration (i.e. the maximal absolute value of coordinates in the target).
Therefore all the numbers in vectors $d_i$ are at most doubly-exponential as well.
Thus also the constant $M$ is at most doubly-exponential.

By an $M$-abstraction of a configuration $c$, denoted by $A_M(c)$, we mean an $\omega$-configuration in $Q\times (\N\cup\set{\omega})^d$ where we substitute in $c$ any value greater or equal $M$ by $\omega$. Concretely:
\[
A_M(c)[j] =
\begin{cases}
c[j] & \text{if } c[j] < M, \\
\omega & \text{if } c[j] \geq M.
\end{cases}
\]

Observe that function $A_M$ has a following property, which is the consequence of the definition of the value $M$,
\begin{equation}\label{eq:abstraction}
c\in \bar{\Set{C}}\iff A_M(c)\in \bar{\Set{C}}.
\end{equation}
Finally, the image of the function $A_M$ is in one to one correspondence with the states of the
$M$-abstraction of $V$, so we will use $M$-abstractions of configurations in both settings: in the VASS $V$ and in the $M$-abstraction of $V$.

Now we can start the actual proof of the lemma. We prove it by contradiction. Namely we assume that from some configuration $c$
there is no accepting run in $V$, but there is an accepting run from the state $A_M(c)$ in the $M$-abstraction of $V$.
Therefore $c\in \bar{\Set{C}}$ and from $A_M(c)$ in the $M$-abstraction it is possible to reach the set of final states $F'$ of the $M$-abstraction. We will show that such $c$ cannot exist.
First observe that $A_M(c)$ cannot be among final states of the $M$-abstraction as then $c\in F$ which contradicts the fact that $c\in \bar{\Set{C}}$.

So for every $A_M(c)$ there is a shortest path in the $M$-abstraction to the set of final states $F'$. Without loss of generality we can assume that $c$ is chosen in such a way that the length of the path from $A_M(c)$ to $F'$ is minimal. We denote this path by $\sigma$.
As $A_M(c)\not\in F'$ we can decompose $\sigma=t\sigma'$ into the head and a tail, where $t=(q,a_t, \delta_t,p)$. Let us consider $\omega$-configuration $c'=A_M(c)+ \eff(t)$ denoting an $\omega$-configuration $(p, v)$ where $v[i]=A_M(c)[i]+\delta_t[i]$ for every $i\in [1,d]$
.
On the one hand we know that $c'\in \bar{\Set{C}}$, as $c \in \bar{\Set{C}}$.
Because of~\eqref{eq:abstraction} we get that $A_M(c')\in \bar{\Set{C}}$. On the other hand in the $M$-abstraction of $V$ the path $\sigma$ from $A_M(c)$ traverses through $A_M(c')$. Therefore the path $\sigma'$ from $A_M(c')$ to $F'$ is shorter than the path $\sigma$. So we get a contradiction with the choice of $c$ as $c'\in \bar{\Set{C}},$ and the path from $A_M(c')$ to $F'$ is shorter than from $c$.
\end{proof}

\subsection{Proof of Lemma~\ref{lem:kafa-to-kdfa}}\label{sec:kafa-to-kdfa}

 We recall the statement of Lemma~\ref{lem:kafa-to-kdfa}.

\vspace{0.3cm}

\noindent
\textbf{Lemma~\ref{lem:kafa-to-kdfa}.}
Let $\A = (\Sigma, Q, T, q, F)$ be a $k$-ambiguous finite automaton for some $k \in \N$.
Let $M$ be a finite monoid and $h: \Sigma^* \to M$ be a homomorphism
recognising all the state languages of the automaton $\A$.
Then the decoration $\A_{(M,h)}$ is a $k$-deterministic finite automaton.

\begin{proof}[Proof of Lemma~\ref{lem:kafa-to-kdfa}.]
Let $w \in \Sigma^*$ be any word and $\decor{w} \in (\Sigma_\eps \times M )^*$ be its decoration.
We have to prove that there are at most $k$ maximal runs over the word $\decor{w}$.
To do so we prove two statements 
\begin{itemize}
\item All maximal runs over the word $\decor{w}$ have the same length.
\item The multiset $X^{\decor{w}}$ of states reachable in $\decor{\A}$ along $\decor{w}$ is of size at most $k$.
\end{itemize}
Indeed, if the above holds, then all the maximal runs have the same length, therefore their number is bounded by the size of $X^{\decor{w}}$ i.e. by $k$ as required.

In the first claim actually we show that either there is only one empty run of length zero or
all the maximal runs are over the whole word $\decor{w}$.
We consider two cases: either $w \not\in L(\A)$ or $w \in L(\A)$. 

We claim that in the first case, the only maximal run has length zero.
Notice that, then the first step should be $(q,\bot) \trans{(\epsilon, h(w))}(q, h(w))$.
According to the definition of $M$ for every $m\in M$ and every state 
$p\in Q$ we have either $h^{-1}(m)\subseteq L(p)$ or $h^{-1}(m)\cap L(p)=\emptyset$. Thus, in our case, it is that $h^{-1}(h(w))\cap L(q)=\emptyset$, because $w \not\in L(\A)$.
This means that the state $(q,h(w))$ is not present in the automaton $\decor{\A}$,
since $\decor{\A}$ does not have states with empty language.
So, indeed, the only maximal run is the empty one.

We prove the second case in a slightly more general setting: for all the prefixes of a decorated word.
The reason is that we show the statement  via an induction on the length of this prefix.
Suppose $u$ is a prefix of $w$ and $\bar{u}$ is the corresponding prefix of $\decor{w}$.
For the empty prefix $u = \eps$ the statement trivially holds.
For an induction step, suppose that $\bar{u} = \bar{u}' \cdot (a,m)$.
The induction hypothesis says that all the maximal runs along $\bar{u}'$ are over the whole word $\bar{u}'$.
Let $\decor{X^{\bar{u}'}} = \multiset{(q_1,m'), (q_2,m'), \ldots, (q_\ell,m')}$ be the multiset of states reached by the runs labelled $\bar{u}'$, for $m' = h(a) \cdot m$. Notice that $a\cdot h^{-1}(m)\subseteq h^{-1}(m')$.
Recall that by the construction of $\decor{\A}$ the languages of states $(q_i, m')$
are not empty, therefore for each $i\in \{1\ldots \ell\}$ we have $h^{-1}(m')\subseteq L(q_i)$.
The two above inclusions imply that $a\cdot h^{-1}(m) \subseteq L(q_i)$.
So, from each state $(q_i,m')\in \decor{X^{\bar{u}'}}$ there is an outgoing transition labelled by $(a,m)$.
Thus, all the maximal runs over $\bar{u}$ are over the whole word $\bar{u}$.

To prove the second claim, we assume that $w$ is accepted by the automaton $\A$ as otherwise
there is just one maximal run of length zero.
Let
$$\decor{X^{\decor{w}}} = \multiset{(q_1,m_1), (q_2,m_1), \ldots, (q_\ell,m_1)}$$
be the multiset of states reached by the runs labelled $\decor{w}$. Furthermore, because of the proof that all the maximal runs have the same length, we conclude that there are exactly $\ell$ runs along $\decor{w}$. Moreover, all of them are accepting, as they end in states with the monoid element $h(\epsilon)$.
Observe that the states along the runs of $\decor{w}$ never differ on the second component, as the
second component is determined by the image $h$ on the suffix of the word. Thus, if two runs over $\decor{w}$
differ, then they differ at the first component and in consequence they induce two different runs along $w$ in $\A$.
For example, such two different runs over $\decor{w}$ with the first difference after the letter $(a_{i+1},m_{i+1})$:
\begin{align*}
(q, \perp) & \trans{(\eps, m)} (q, m) \trans{(a_1, m_1)} (p_1, m_1) \trans{(a_2, m_2)} \ldots (p_i,m_i)\trans{(a_{i+1},m_{i+1})}(p_{i+1},m_{i+1})\ldots\\
& \trans{(a_{n-1}, m_{n-1})} (p_{n-1},m_{n-1}) \trans{(a_n, m_{n})} (p_n, h(\epsilon)),
\end{align*}
\begin{align*}
(q, \perp) & \trans{(\eps, m)} (q, m) \trans{(a_1, m_1)} (p_1, m_1) \trans{(a_2, m_2)} \ldots (p_i,m_i)\trans{(a_{i+1},m_{i+1})}(r_{i+1},m_{i+1})\ldots\\
& \trans{(a_{n-1}, m_{n-1})} (r_{n-1},m_{n-1}) \trans{(a_n, m_{n})} (r_n, h(\epsilon)).
\end{align*}
induce the following two different runs over $w$:
\begin{align*}
 q \trans{a_1} p_1 \trans{a_2} \ldots p_i\trans{a_{i+1}}p_{i+1}\ldots  \trans{a_{n-1}} p_{n-1} \trans{a_n} p_n,
\end{align*}
\begin{align*}
q \trans{a_1} p_1 \trans{a_2} \ldots p_i\trans{a_{i+1}}r_{i+1}\ldots \trans{a_{n-1}} r_{n-1} \trans{a_n} r_n.
\end{align*}
Therefore, as $\A$ is $k$-ambiguous and $\A$ has at least $\ell$ accepting runs then 
we necessarily have $\ell \leq k$, which finishes the proof that size of $X^{\bar{u}}$ is at most $k$.
\end{proof}

\section{Future research}\label{sec:future}

\myparagraph{VASS accepting by configuration}
In our work we prove Theorem~\ref{thm:complement-bavass}
stating that for a $k$-ambiguous upward-VASS
one can compute a downward-VASS recognising the complement of its language.
This theorem implies all our upper bound results, namely
decidability of language inclusion of an upward-VASS in a $k$-ambiguous upward-VASS
and language equivalence of $k$-ambiguous upward-VASS.
The most natural question which can be asked in this context is whether Theorem~\ref{thm:complement-bavass}
or some of its consequences generalises to singleton-VASS
(so VASS accepting by a single configuration) or more generally to downward-VASS.
Our results about complementing
deterministic VASS apply also to downward-VASS. However, generalising
our results for nondeterministic (but $k$-ambiguous or unambiguous) VASS encounters essential barriers. Techniques from Section~\ref{sec:uvass} do not
work as the regular-separability result from~\cite{DBLP:conf/concur/CzerwinskiLMMKS18}
applies only to upward-VASS. Techniques from Section~\ref{sec:bavass}
break as the proof of Lemma~\ref{lem:ba-control} essentially uses the fact that
the acceptance condition is upward-closed. Thus it seems that one would need to develop novel
techniques to handle the language equivalence problem for unambiguous VASS
accepting by configuration.

\myparagraph{Weighted models}
Efficient decidability procedures for language equivalence were obtained
for finite automata and for register automata using
weighted models~\cite{DBLP:journals/iandc/Schutzenberger61b,DBLP:conf/lics/BojanczykKM21}.
For many kinds of systems one can naturally define weighted models by adding weights
and computing value of a word in the field $(\Q,+,\cdot)$.
The decidability of equivalence for weighted models easily implies
language equivalence for unambiguous models, as accepted words always have the output equal to one,
while rejected words always have the output equal zero. Thus one can pose a natural
conjecture that decidability of language equivalence for unambiguous models always comes as a byproduct
of equivalence of the weighted model. Our results show that this is however not always the case
as VASS are a counterexample to this conjecture.
In the case of upward-VASS language equivalence for unambiguous models is decidable.
However, equivalence for weighted VASS is undecidable, as it would imply decidability of path equivalence
(for each word, both systems need to accept by the same number of accepting runs),
which is undecidable for VASS~\cite{DBLP:journals/tcs/Jancar01}.

\myparagraph{Unambiguity and separability}
Our result from Section~\ref{sec:uvass} uses the notion of regular-separability in order to obtain
a result for unambiguous VASS. This technique seems to generalise to some other well-structured transition systems.
It is natural to ask whether there is a deeper connection between the notions of separability and unambiguity
that can be explored in future research.

\section*{Acknowledgment}

We thank Filip Mazowiecki for asking the question for boundedly-ambiguous VASS and formulating the conjecture that control automata of boundedly-ambiguous VASS can be made boundedly-ambiguous. We also thank him and David Purser for inspiring discussions on the problem. We thank Thomas Colcombet for suggesting the way of proving Theorem~\ref{thm:reg-sep-upgrade}, Mahsa Shirmohammadi for pointing us to the undecidability result~\cite{DBLP:journals/tcs/Jancar01} and Lorenzo Clemente for inspiring discussions on weighted models. We thank Petr Jan\v{c}ar for many helpful remarks and simplifying the proofs of Proposition~\ref{prop:blowup}~and~Lemma~\ref{lem:threshold}. We also thank Marin Ricros for his detailed comments useful in improving our paper.

\bibliographystyle{alphaurl}
\bibliography{citat}

\end{document}